\documentclass[12pt, draftclsnofoot,onecolumn]{IEEEtran}
\usepackage{amssymb}
\usepackage[cmex10]{amsmath}
\usepackage{stfloats}
\usepackage{graphicx}
\usepackage{caption}
\usepackage{subcaption}
\usepackage{tabularx}
\usepackage{epsfig,epsf,color,balance,cite}
\usepackage{verbatim}
\usepackage{url}
\usepackage{bm}
\usepackage{ragged2e}

\usepackage{booktabs} 
\usepackage{tikz, pgfplots,graphicx,xcolor}
\usetikzlibrary{plotmarks,spy,backgrounds}
\pgfplotsset{compat=newest}

\usepackage{mathrsfs}
\usepackage{amsthm}
\usepackage{cleveref}
\usepackage{mathtools}
\usepackage{extarrows}
\DeclarePairedDelimiter{\ceil}{\lceil}{\rceil}

\usepackage{mathbbol}

\def\mindex#1{\index{#1}}

\def\sq{\hbox{\rlap{$\sqcap$}$\sqcup$}}
\def\qed{\ifmmode\sq\else{\unskip\nobreak\hfil
\penalty50\hskip1em\null\nobreak\hfil\sq
\parfillskip=0pt\finalhyphendemerits=0\endgraf}\fi\medskip}

\long\def\defbox#1{\framebox[.9\hsize][c]{\parbox{.85\hsize}{%
\parindent=0pt
\baselineskip=12pt plus .1pt      
\parskip=6pt plus 1.5pt minus 1pt 
 #1}}}

\long\def\beginbox#1\endbox{\subsection*{}%
\hbox{\hspace{.05\hsize}\defbox{\medskip#1\bigskip}}%
\subsection*{}}

\def\endbox{}

\newsavebox{\junk}
\savebox{\junk}[1.6mm]{\hbox{$|\!|\!|$}}

\def\det{{\mathop{\rm det}}}

\def\bfmath#1{{\mathchoice{\mbox{\boldmath$#1$}}%
{\mbox{\boldmath$#1$}}%
{\mbox{\boldmath$\scriptstyle#1$}}%
{\mbox{\boldmath$\scriptscriptstyle#1$}}}}

\def\bfmY{\bfmath{Y}}

\def\bfmhhaY{\bfmath{\hhaY}} 
\def\bfmhhaY{\hbox to 0pt{$\widehat{\bfmY}$\hss}\widehat{\phantom{\raise 1.25pt\hbox{$\bfmY$}}}}

\def\til={{\widetilde =}}

 \def\FRAC#1#2#3{\genfrac{}{}{}{#1}{#2}{#3}}

\def\ddtp{{\mathchoice{\FRAC{1}{d^{\hbox to 2pt{\rm\tiny +\hss}}}{dt}}%
{\FRAC{1}{d^{\hbox to 2pt{\rm\tiny +\hss}}}{dt}}%
{\FRAC{3}{d^{\hbox to 2pt{\rm\tiny +\hss}}}{dt}}%
{\FRAC{3}{d^{\hbox to 2pt{\rm\tiny +\hss}}}{dt}}}}

\def\average#1,#2,{{1\over #2} \sum_{#1}^{#2}}

\def\eye(#1){{\bf(#1)}\quad}

\newtheorem{theorem}{{\bf Theorem}}

\newtheorem{remark}{{\bf Remark}}

\newtheorem{lemma}{{\bf Lemma}}

\def\eq#1/{(\ref{e:#1})}

\newcommand{\beqn}[1]{\notes{#1}%
\begin{eqnarray} \elabel{#1}}

\newcommand{\eeqn}{\end{eqnarray} }

\newcommand{\beq}[1]{\notes{#1}%
\begin{equation}\elabel{#1}}

\newcommand{\eeq}{\end{equation}}

\def\bdes{\begin{description}}
\def\edes{\end{description}}

\newcounter{rmnum}

\newcounter{anum}

{\end{list}}

\def\ass(#1:#2){(#1\ref{#1:#2})}

\def\ritem#1{
\item[{\sf \ass(\current_model:#1)}]
}

\newenvironment{recall-ass}[1]{%
\begin{description}
\def\current_model{#1}}{
\end{description}
}

\long\def\comment#1{}

\newfont{\bb}{msbm10 scaled 1100}

\newcommand{\EE}{\mbox{\bb E}}

\newcommand{\hv}{{\bm h}}

\newcommand{\Em}{{\bm E}}

\newcommand{\Hm}{{\bm H}}
\newcommand{\Id}{{\bm I}}

\newcommand{\Xm}{{\bm X}}
\newcommand{\Ym}{{\bm Y}}
\newcommand{\Zm}{{\bm Z}}

\newcommand{\Cc}{{\cal C}}

\newcommand{\Jc}{{\cal J}}
\newcommand{\Kc}{{\cal K}}

\newcommand{\Nc}{{\cal N}}

\newcommand{\Tc}{{\cal T}}

\renewcommand{\det}{{\hbox{det}}}

\renewcommand{\arg}{{\hbox{arg}}}

\newcommand{\SNR}{{\sf snr}}

\newcommand{\herm}{{\sf H}}

\newcommand{\transp}{{\sf T}}

\newcommand{\mk}{{\rm -\!o\!-}}

\usepackage{algorithm}
\usepackage{algpseudocode}
\usepackage{graphicx}
\usepackage{epstopdf}
\begin{document}

\title{Achievable Rates for a Distributed Antenna System with No Channel State Information at the Central Processor}

\author{
	\IEEEauthorblockN{Yi Song, \emph{Student Member, IEEE},
	Hao Xu, \emph{Senior Member, IEEE},
        Kai Wan, \emph{Member, IEEE}, 
		Kai-Kit Wong, \emph{Fellow, IEEE},
		Giuseppe Caire, \emph{Fellow, IEEE},
		and
		Shlomo Shamai (Shitz), \emph{Life Fellow, IEEE}
	}
	\thanks{
	The work of Y. Song and G. Caire was supported by BMBF Germany in the program of ``Souverän. Digital. Vernetzt.'' Joint Project 6G-RIC (Project IDs 16KISK030).
	The work of H. Xu was supported by the Fundamental Research Funds for the Central Universities under grant 2242025R10001.
The work of K.~Wan was  partially funded by the National Natural Science Foundation of China (NSFC-12141107),  the Key Research and Development Program of Wuhan under Grant 2024050702030100, and Wuhan ``Chen
Guang'' Pragram under Grant 2024040801020211.
The work of S. Shamai has been supported by the German Research
Foundation (DFG) via the German-Israeli Project Cooperation (DIP), under Project SH 1937/1-1.

Y. Song and G. Caire are with the Faculty of Electrical Engineering and Computer Science at the Technical University of Berlin, 10587 Berlin, Germany (e-mail: yi.song@tu-berlin.de; caire@tu-berlin.de).
    
H. Xu is with the National Mobile Communications Research Laboratory, Southeast University, Nanjing 210096, China (e-mail: hao.xu@seu.edu.cn).

K.-K. Wong is with the Department of Electronic and Electrical Engineering, University College London, WC1E 7JE London, U.K., and also with the Department of Electronic Engineering, Kyung Hee University, Yongin-si, Gyeonggi-do 17104, Republic of Korea (e-mail: kai-kit.wong@ucl.ac.uk).

K. Wan is with School of Electronic Information and Communications, Huazhong University of Science and Technology, Wuhan, China (email: kai\_wan@hust.edu.cn).

S. Shamai (Shitz) is with the Viterbi Electrical Engineering Department, Technion Israel Institute of Technology, Haifa 32000, Israel (e-mail: sshlomo@ee.technion.ac.il).

Corresponding author: Hao Xu.
	}
}

\maketitle

\begin{abstract}
A recent trend in wireless communications considers the migration of traditional monolithic base stations to
the so-called ``disaggregated architecture'', where radio units (RUs) implement only the low-level physical layer functionalities such as demodulation, and A/D conversion, while the high-level physical layer, such as channel decoding, is
implemented as software-defined functions running on general-purpose hardware 
in some remote central processing unit (CP).  The corresponding information theoretic model for the uplink (from the wireless users to the CP) is a multiaccess-relay channel with {\em primitive oblivious} relays.
The relays (RUs) are ``oblivious'', as they are agnostic of the users' codebooks, and ``primitive'',  since the {\em fronthaul} 
links (from RUs to CP) are error-free with limited capacity. This class of networks has been intensely studied in the information theoretic literature, where several approximated or exact (under certain conditions) capacity results have been derived. 
In particular, in the Gaussian case, the model has been analyzed for fixed and known channel state. This paper is motivated by the fact that, in practice,  the channel state is a random process, and it is estimated at the base station side 
through uplink pilot symbols sent by the users.  The pilot dimension may take up a large portion 
of the channel coherence block, i.e.,  the number of symbols over which the channel state remains approximately constant. 
Hence, sending both pilot and data symbols from the relays to the CP may require a significant overhead, 
especially when the fronthaul capacity is small. 
As a prototypical problem,  we consider the ergodic achievable rate for a ``diamond'' network formed by a 
single user  and two relays where the channel state is known at the relays, 
but not known at the CP. Despite its simplicity, it turns out that an exact characterization of the 
ergodic capacity for this model is surprisingly difficult. Then, we develop an analytical upper bound
and methods to numerically evaluate such upper bound. 
In order to obtain lower bounds, we resort to specific  analytically/numerically tractable achievability strategies. When designing such strategies, we need to take into account that the CP has no channel state information and that each relay has only statistical knowledge of the channel state of the other relay. 
Under these constraints, we propose some achievable schemes based on low-dimensional oblivious processing at the relays. Next, we consider a more challenging case with two users and two relays.
For this case, we develop an analytical upper bound and discuss its numerical evaluation. It turns out that the achievable 
schemes for the single-user case cannot be extended to the two-user case. Hence, we propose a new achievable scheme for the two-user case by jointly encoding the received signal and the channel state. Simulation results show that the proposed achievable schemes in both cases perform close to the upper bound across a broad range of system parameters. 
\end{abstract}


\IEEEpeerreviewmaketitle

\section{Introduction}
\label{introduction}

Traditional Radio Access Network (RAN) architectures in the 2nd to 4th generation of wireless systems 
are based on a monolithic building block, physically implemented in the base station sites. 
With the advent of the 5th generation (5G), a more flexible, scalable,  
and cost-effective ``disaggregated'' RAN was promoted  \cite{3gpp2018technical, ahmadi20195g, xu2025distributed}. 
This new RAN is based on functional splits and identifies two basic building blocks, 
Radio Units (RUs) and Decentralized Units (DUs).\footnote{The so-called 3GPP disaggregated RAN architecture, 
specifies also a third basic building block referred to as Centralized Unit (CU), 
which implements non-realtime functions of network orchestration, resource allocation, etc. However, 
the CU is located in the telecom operator network (cloud), and is not attached to the fronthaul.
Hence, CUs are irrelevant for the problem treated in this paper.}
RUs implement the lower physical layer functionalities 
(RF front-end, demodulation, A/D conversion). Clusters of several 
RUs are connected to a DU, that implements the upper physical layer functionalities
(multiple-input and multiple-output (MIMO) precoding, channel decoding) as well as some upper layer functions, up to the interface with the {\em core network}. 
A central idea of the disaggregated RAN is that the RUs can operate in an oblivious way, agnostic of 
the users' codebooks, while the decoding functions that require this knowledge can be implemented as software-defined network functions in the DU and run on general-purpose hardware.\footnote{For this reason, the 3GPP disaggregated RAN is sometimes referred to as C-RAN (Cloud-RAN), since the virtualized functions are run ``in the cloud'' \cite{7018201, 6897914}.}
RUs are connected with the DU via a {\em fronthaul network}, usually implemented by  
point-to-point links, not interfering with the wireless access segment, i.e., the users-to-RUs channel. 

A popular information theoretic model motivated by the uplink of a disaggregated RAN architecture consists of
a multiaccess wireless relay channel formed by $L$ users, $K$ oblivious relays (i.e., the RUs),  
and a  central processor (CP) \cite{park2014fronthaul}.  The first hop, from the users to the relays, is modeled as a Gaussian interference network. 
The second hop (the fronthaul) is formed by a collection of non-interfering capacity-limited links (see Fig.~\ref{CRAN}).
The relays are oblivious in the sense that they are agnostic of the users' codebooks 
\cite{homri2018oblivious,katz2019gaussian,aguerri2019TIT}.
This corresponds to the fact that the low-level functions implemented in the RUs need to know only the signal format
(carrier, bandwidth, time-frequency multiplexing), but are agnostic of the specific modulation and coding 
used by each user to encode its information message. 
Two-hop relay networks where the second hop (from relays to end receiver) is formed by error-free non-interfering 
links with limited capacities are usually referred to as ``primitive'' relay networks 
\cite{kim2008capacity, katz2019gaussian,katz2021filtered,xudistributed,simeone2016cloud}.
This class of networks has been intensively investigated from an information theoretic viewpoint (see
Section~\ref{sec:relatedworks}). 

\begin{figure}
\centering
\begin{tikzpicture}[node distance = 0.03\textwidth]
\tikzstyle{neuron} = [circle, draw=black, fill=white, minimum height=0.05\textwidth, inner sep=0pt]
\tikzstyle{rect} = [rectangle, rounded corners, minimum width=0.05\textwidth, minimum height=0.05\textwidth,text centered, draw=black, fill=white]
    \node [neuron] (neuron1) {\small{User 1}};
    \node [below of=neuron1, yshift=-0.05\textwidth, neuron] (neuron2) {\small{User 2}};
    \node [below of=neuron2, yshift=-0.01\textwidth] (neuron3) {$...$};
    \node [below of=neuron3, yshift=-0.01\textwidth, neuron] (neuron4) {\small{User $L$}};
     \node [above of=neuron1, xshift=0.2\textwidth, yshift=0.05\textwidth, neuron] (add1) {$+$};
      \node [above of=add1, yshift=0.05\textwidth] (N1) {$N_1$};
      \node [right of=neuron2, xshift=0.18\textwidth, neuron] (add2) {$+$};
      \node [below of=add2, yshift=-0.01\textwidth] (add3) {$...$};
      \node [below of=add3, yshift=-0.1\textwidth, neuron] (add4) {$+$};
      \node [above of=add4, yshift=0.05\textwidth] (N3) {$N_K$};
      \draw [->,line width=1pt] (neuron1) --(add2);
        \draw [->,line width=1pt] (neuron2) -- (add2);
     \draw [->,line width=1pt] (neuron1) -- (add1);
          \draw [->,line width=1pt] (neuron2) -- (add1);
         \node [above of=add2, yshift=0.05\textwidth] (N2) {$N_2$};
    \node [right of=add1, xshift=0.15\textwidth, rect] (R1) {\small{Relay 1}};
     \node [right of=add2, xshift=0.15\textwidth, rect] (R2) {\small{Relay 2}};
     \node [below of=R2, yshift=-0.02\textwidth] (R3) {$...$};
         \node [left of=add4, xshift=0.22\textwidth, rect] (R4) {\small{Relay $K$}};
        \node [right of=R2, xshift=0.2\textwidth, rect] (D) {\small{CP}};
     \draw [->,line width=1pt] (N1) -- (add1);
      \draw [->,line width=1pt] (N3) -- (add4);
       \draw [->,line width=1pt] (neuron1) -- (add4);
         \draw [->,line width=1pt] (neuron2) -- (add4);
            \draw [->,line width=1pt] (neuron4) -- (add4);
             \draw [->,line width=1pt] (neuron4) -- (add1);
              \draw [->,line width=1pt] (neuron4) -- (add2);
     \draw [->,line width=1pt] (add1) -- (R1);
       \draw [->,line width=1pt] (R1) -- (D);
       \draw [->,line width=1pt] (N2) -- (add2);
               \draw [->,line width=1pt] (add4) -- (R4);
         \draw [->,line width=1pt] (add2) -- (R2);
         \draw [->,line width=1pt] (R2) -- (D);
                \draw [->,line width=1pt] (R4) -- (D);
\end{tikzpicture}
\caption{A disaggregated RAN model consisting of $L$ users, $K$ RUs (relays) and a centralized processor.}  
\label{CRAN}
\end{figure}
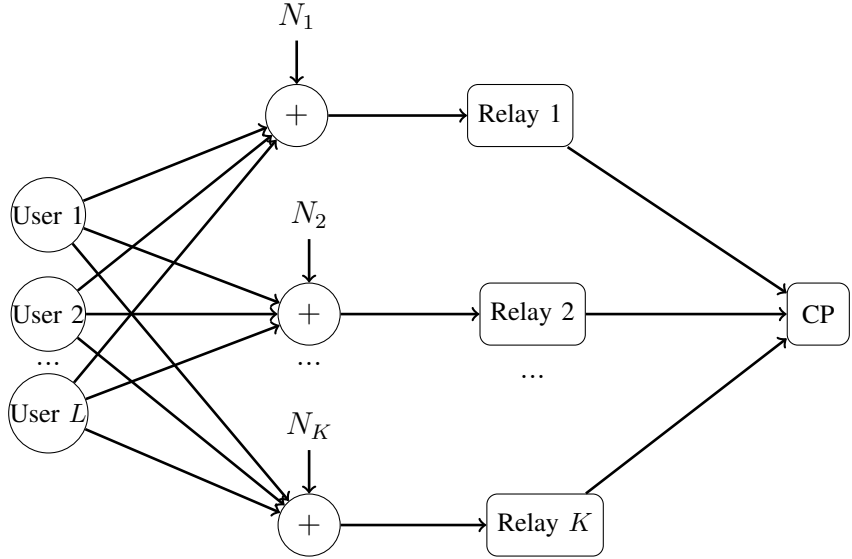


In order to motivate our work, it is important to point out that 
in the Gaussian case (relevant for wireless communication), the model has been analyzed for 
fixed and known channel state, i.e., the matrix of channel coefficients between the users' transmit antennas and 
the relays' receiving antennas is fixed and known to all. 
As a matter of fact, the channel state is a matrix-valued random process that stays approximately constant over 
coherence blocks of $T$ signal dimension in the time-frequency domain. In practical systems, 
in the uplink, users send pilot symbols in each coherence block, to allow the receiver to estimate the 
channel and apply a {\em coherent detection} scheme. {More fundamentally, 
from \cite{zheng2002communication} we know that
the high signal-to-noise ratio (SNR) capacity of a block-fading MIMO channel with ${L}$ transmit antennas, 
$K$ receive antennas, and coherence block length $T$, is given by 
$(1 - \frac{M^*}{T}) M^* \log \SNR + o(\log \SNR)$ where $M^* = \min\{L, K, \frac{T}{2}\}$.} 
From a simple cutset bound argument, this is 
also an upper bound to the sum capacity of a Gaussian multiaccess channel with $L$ single-antenna users, 
a receiver with $K$ antennas, and block-fading with coherence block $T$. In addition, the pre-log factor 
$(1 - \frac{M^*}{T}) M^*$ is achievable by using an explicit uplink pilot scheme letting $M^*$ users transmit mutually orthogonal pilots of dimension $M^*$ in each channel coherence block.
When $L$ and $K$ are comparable with $T$, the pre-log factor is maximized by letting 
$M^* = \frac{T}{2}$, i.e., the number of active users (and the corresponding pilot dimension) 
is equal to half of the channel coherence block.\footnote{In modern massive MIMO
we may have hundreds of antennas and $T$ may be as small as $200$ symbols, i.e., the time-frequency
dimension of a so-called ``resource block'' in the 5GNR standard \cite{3gpp2018technical}.}
The above argument shows that it is not unlikely that wireless systems may operate in a regime where 
the pilot dimension is a large fraction of the whole channel coherence block length.

The above high-SNR capacity result holds when the front-end rate between the receiver antennas 
and the processing unit is unbounded, i.e., the decoder has direct access to the signals received 
at the $K$ antenna ports. However, in a primitive relay network, the relays are connected to the central processor (CP) via capacity-limited links.
In this case, it is not known if the uplink piloting scheme is asymptotically optimal. Even insisting on the traditional pilot scheme (motivated by the practical implementation of wireless systems), when the  pilot dimension is a significant fraction of the coherence block, sending (through some oblivious processing, e.g., quantization) the received pilot field through the fronthaul may incur a large penalty. 

To illustrate this point, consider the simple case of $L = 1, K = 1$ where the channel
between the user and the relay is given by $Y = e^{j\phi} X + N$, where $N \sim \Cc\Nc(0,1)$, $\phi \sim $ Uniform$[0,2\pi]$ and 
$\EE[|X|^2] = \SNR$. Suppose that the random phase $\phi$ remains constant for blocks of $T = 2$ consecutive symbols, 
and is  independent and identically distributed (i.i.d.) across the blocks. A fixed and designed a prior pilot-based scheme sends $X = X_0$ as the first reference symbol of each block, and an information symbol $X$ as the second symbol of each block.  A naive fronthaul strategy would quantize and transmit through the fronthaul $2$ symbols per block. 
However, a simple oblivious processing at the relay may do better by sending 
only the second symbol in the block after phase de-rotation by the first reference symbol 
(one symbol per block).  In short,  when the fronthaul rate is a limiting factor and the pilot dimension is 
significant with respect to the coherence block, the oblivious relay network become highly non-trivial and, despite the many existing results summarized in Section \ref{sec:relatedworks},  the problem of determining the capacity region, or at least tight scaling laws of the sum  capacity, is wide open. 

{In order to make progress, in this work we bypass the pilot-based channel estimation issue
and consider a model where the channel state is given {\em exactly} at the relays 
but completely unknown to the CP. In addition, the channel state is i.i.d. across time.
Therefore, for each received signal, there is one independent channel coefficient at each relay. This models the fact that the signal dimension and the channel state dimension are comparable. 
The problem then falls in the class of networks studied in \cite{aguerri2019TIT}. 
However,  the capacity expression in \cite{aguerri2019TIT} is given for discrete memoryless channels. 
For the Gaussian case, achievable rates are given under the assumption of Gaussian inputs 
for fixed channel matrix known to both the relays and the CP. 
This generalizes to the case of random i.i.d. channel state known to all, 
but the actual computation of an ergodic achievable rate region or even just the sum rate is highly non-trivial (see Section \ref{problem_formu}).
In addition, when the CP does not know the channel state, the general capacity expression in \cite{aguerri2019TIT} still holds, 
{but  even evaluating achievable inner bounds becomes highly non-trivial, 
due to the unclear choice of certain auxiliary random variables (see discussion in Section \ref{achive_shmes_2}}).}

The problem at hand is also closely related to the so-called Information Bottleneck (IB) problem 
\cite{tishby2000information,hassanpour2017overview, shwartz2017opening, zaidi2020information, goldfeld2020information, xudistributed}.
In particular, it is well-known that in the {\em Gaussian IB} problem $X \mk Y \mk Z$, 
when the source $X$ and the observation $Y$ are jointly Gaussian, then  
$I(X;Z)$ is maximized subject to the bottleneck constraint $I(Y; Z) \leq C$ by letting 
$Z$ jointly Gaussian with $Y$ \cite{chechik2003information}. However, in the case where $Y = SX + N$ with $X, N, S$ independent Gaussian, the optimal choice of $Z$ in the bottleneck problem
$X \mk (Y,S) \mk Z$ is unknown.  Hence, even for this elementary model, the 
IB problem is non-trivial and numerically hard (e.g., no cardinality bounds). 

Instead of pursuing uncomputable/numerically intractable rate region expressions, 
in this paper our objective is to provide simple upper and lower bounds
for the achievable capacity in two simple cases of the general problem, 
with one and two users, and two relays, that can be relatively easily numerically evaluated.

\subsection{Related Work}  \label{sec:relatedworks}

{ 
The disaggregated RAN (or C-RAN) model with oblivious processing for a single user and two relays was first studied in \cite{dandervoich2008communication}. Sanderovich \emph{et. al} in \cite{dandervoich2008communication} derived an achievable rate and an upper bound on the model with discrete memoryless channels between the user and relays, and further extended the theoretical results to Gaussian channels. 
The converse bound was later found in \cite{estella2018distributed,aguerri2019distributed2} for both channels. 
{Interestingly, the single-user model was found to share the same capacity region as the Chief Executive Officer (CEO) problem under logarithmic loss \cite{courtade2013multiterminal}}. Furthermore, Aguerri \emph{et. al} in \cite{aguerri2019TIT} extended the capacity region to the general $K$-user case for both discrete memoryless and Gaussian vector channels (the latter 
subject to Gaussian inputs).

The aforementioned works \cite{aguerri2019TIT, dandervoich2008communication, estella2018distributed, aguerri2019distributed2} assume fixed channel states known to all nodes. In \cite{sanderovich2009distributed}, fast fading channels were considered, but under the assumption that the channel state was also perfectly known to the CP. Although this assumption may be hold approximately when the channel coherence block is very large with respect to the number of users, it is generally over-optimistic when the coherence block is of the same order of the number of users.  Caire \emph{et. al} in \cite{caire2018information} simplified the problem for the single user/single relay case 
by assuming that the relay has perfect knowledge of the channel state but the CP does not. These results where extended by 
Xu \emph{et. al} in \cite{IBxu} to the MIMO single user/single relay, where both the user and the relay have 
multiple antennas.  In \cite{xudistributed, song2023distributed},  we considered 
the disaggregated RAN model for single user and two relays under Raleigh fading channel for the single antenna and MIMO cases under different assumptions on the channel state knowledge. 
These papers can be seen as preliminary to this present paper.
}

\subsection{Contributions}  \label{sec:contributions}

In this paper, we study the single-user ($L = 1$) and two-user ($L = 2$) scenarios of the fronthaul-constrained disaggregated RAN model shown in Fig.~\ref{CRAN}, with $K = 2$ relays with i.i.d. Rayleigh fading channels.  The channel state is known at the relays (genie-aided) but not known at the CP. 
As said before, our focus is to derive simple and analytical upper bounds and achievable lower bounds on the 
ergodic (sum) rate under Gaussian inputs.\footnote{Sometimes this assumption is referred to as ``point-to-point coding''
in the sense  that users are restricted to use random Gaussian codebooks as if they were operating in a 
single-user (point-to-point) Gaussian channel \cite{baccelli2011interference}. 
This is regarded as an information theoretic version of the fact, dictated by practice, that in real systems
the channel codes are designed independently of the network scenario, and optimized for the single-user Gaussian channel.}

\begin{itemize}
    \item \textbf{Single-user case:} we start with the results in \cite{dandervoich2008communication} for fixed channel states and notice that the case of random channel state known only to the relays yields an intractable expression for the sum capacity. Then, we derive a rate upper bound by assuming that the CP also has the access to the channel state.
    The calculation of this upper bound is still highly not trivial since it involves 
    an optimization over a space of functions. Then, we provide a relaxed upper bound that can be analytically calculated, 
    and a lower bound which can be calculated by solving a convex optimization problem 
    which traps the original (tighter) upper bound. We also provide a 
    numerical method based on stochastic optimization via the Lyapunov Drift-Plus-Penalty (DPP) approach  \cite{georgiadis2006resource, neely2010stochastic}
    to approximate the upper bound. 

    Taking into account the fact that each relay has access only to its own channel channel state and has only statistical knowledge 
of the other relay's channel,  we propose three achievable schemes based on local relay oblivious 
processing.
    
    \item \textbf{Two-user case:} Based on the result in \cite{aguerri2019TIT} for fixed channel states, we derive the sum capacity for this model. However, a computable expression for the sum capacity appears to be a very hard 
    open problem, since it is not clear how to choose the involved auxiliary random variables. 
Then, we propose an upper bound that can be analytically computed, 
by assuming that relays can cooperate and that the CP has also perfect channel state information.
    
In contrast to the single-user case, the relay strategies based on local processing (e.g., 
    channel inversion or MMSE estimation) are either infeasible or highly suboptimal in the two-user setting. 
Taking into account the fact that each relay has access only to its own channel channel state and has only statistical knowledge 
of the other relay's channel, we propose an achievable scheme where each relay jointly compresses both its 
received signal and its channel state.
    
    \item We verify numerically that the gaps between the analytical upper bounds and the achievable rates 
    obtained proposed achievable schemes are quite small in several regime of interest, demonstrating the remarkable performance of the simple achievable schemes across a wide range of system parameters.
\end{itemize}

\subsection{Notations}

The real and complex fields by $\mathbb{R}$ and $\mathbb{C}$, respectively. 
We use boldface uppercase letters to indicate matrices, boldface lowercase letters for vectors, and 
calligraphic uppercase letters for sets. $\Id_K$ is the $K \times K$ identity matrix and $\mathbf{0}$ is 
the all-zero vector or matrix. $(\cdot)^\herm$ denotes conjugate transpose, and $\mathbb{E}[\cdot]$ is the expectation operation.
The superscript $(\cdot)^c$ for a set 
denotes the complement set. $\cdot \setminus \cdot$ denotes set subtraction operation, e.g., $\Kc \setminus k$ denotes the set $\Kc$ after eliminating the element $k$. To indicate a set of indexed variables or functions $\{a_k : k \in \Kc\}$ 
we use $a_{\Kc}$ or simply $\{a_k\}$ when the index set $\Kc$ is clear from the context.

\section{The Single-user Case}
\label{problem_formu}


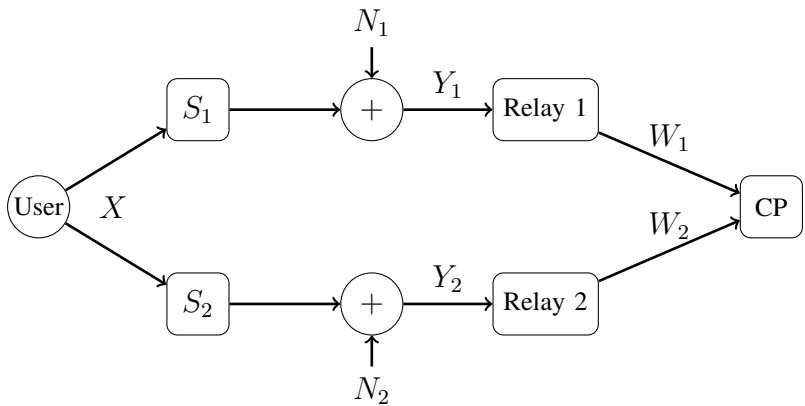
\begin{figure}[ht!]
\centering
\centering
\begin{tikzpicture}[node distance = 0.04\textwidth, auto]
\tikzstyle{neuron} = [circle, draw=black, fill=white, minimum height=0.05\textwidth, inner sep=0pt]
\tikzstyle{rect} = [rectangle, rounded corners, minimum width=0.05\textwidth, minimum height=0.05\textwidth,text centered, draw=black, fill=white]
    \node [neuron] (neuron) {\small{User}};
    \node [right of=neuron, xshift=0.02\textwidth] (X) {$X$};
    \node [above right of=neuron,xshift=0.1\textwidth, yshift=0.05\textwidth, rect] (S1) {$S_1$};
    \node [below right of=neuron, xshift=0.1\textwidth, yshift=-0.05\textwidth, rect] (S2) {$S_2$};
     \node [right of=S1, xshift=0.1\textwidth, neuron] (add1) {$+$};
      \node [above of=add1, yshift=0.03\textwidth] (N1) {$N_1$};
      \node [right of=S2, xshift=0.1\textwidth, neuron] (add2) {$+$};
         \node [below of=add2, yshift=-0.03\textwidth] (N2) {$N_2$};
    \node [right of=add1, xshift=0.1\textwidth, rect] (R1) {\small{Relay 1}};
     \node [right of=add2, xshift=0.1\textwidth, rect] (R2) {\small{Relay 2}};
        \node [right of=neuron, xshift=0.55\textwidth, rect] (D) {\small{CP}};
     \draw [->,line width=1pt] (neuron) -- (S1);
     \draw [->,line width=1pt] (S1) -- (add1);
     \draw [->,line width=1pt] (N1) -- (add1);
     \draw [->,line width=1pt] (add1) -- node[above] {$Y_1$}(R1);
       \draw [->,line width=1pt] (R1) -- node[above] {$W_1$}(D);
     \draw [->,line width=1pt] (neuron) -- (S2);
      \draw [->,line width=1pt] (S2) -- (add2);
       \draw [->,line width=1pt] (N2) -- (add2);
         \draw [->,line width=1pt] (add2) -- node[above] {$Y_2$}(R2);
         \draw [->,line width=1pt] (R2) -- node[above] {$W_2$}(D);
\end{tikzpicture}
\caption{A disaggregated RAN model with a single user and two relays.}  
\label{Block_diagram}
\end{figure}

{
We consider a memoryless Gaussian disaggregated RAN model with oblivious relay processing, as shown in Fig.~\ref{Block_diagram}. In this model, a user communicates with the CP through two relay nodes. 
Each relay node \( k\in \{1,2\} \) connects to the CP via an error-free fronthaul link of finite capacity \( C_k \).  
The relays operate under an oblivious processing constraint, i.e., they lack knowledge of the users' codebooks. This obliviousness is modeled via randomized coding \cite{dandervoich2008communication}, where the user selects codebooks at random without informing the relays, while the CP is aware of the selected codebook.
The user transmits codewords \( X^n \) over Gaussian channels with i.i.d. Rayleigh fading. 
Each relay \( k \) has perfect knowledge of its own channel state sequence $S^n_k$. 
The received signal $(Y_k^n,S_k^n)$ at each relay follows the memoryless channel model:  
\[
    P_{Y_k^n,S_k^n| X^n}(y_k^n,s_k^n|x^n) =\prod_{i=1}^{n} P_{Y_k|X,S_k}(y_{k, i}|x_i,s_{k,i})P_{S_k}(s_{k,i}),
\]
where $P_{Y_k|X,S_k}$ is given by the single-letter model
\begin{equation}\label{obser}
Y_k = S_k X + N_k.
\end{equation}  
Here, \( X \sim {\cal CN}(0, 1) \) represents the channel input, \( S_k \sim {\cal CN}(0, 1) \) 
is the Rayleigh fading coefficient, and \( N_k \sim {\cal CN}(0, \sigma^2) \) denotes additive Gaussian noise at relay \( k \). The channel SNR for relay $k$ is defined as ${\rm SNR} = \frac{1}{\sigma^2}$.
 
In general, each relay collects blocks of length $n$ of the received signal and corresponding channel state $(Y^n_k,S^n_k)$ and encodes them into quantization index $W_k$ in a set of cardinality $2^{n R^{\rm enc}_k}$ 
such that $R^{\rm enc}_k \leq C_k$. The CP, in the presence of the relay messages $(W_1, W_2)$, must recover the information message sent by the user. 

The capacity (under oblivious relay processing and fixed input distribution) of the model in Fig.~\ref{Block_diagram}
is given in \cite[Theorem~$1$]{dandervoich2008communication} and takes on the form
\begin{subequations}\label{IB_problem}
	\begin{align}
	\mathop {\max }\limits_{\{P_{Z_k| Y_k, S_k}: k\in \Kc\}} \;& I(X; Z_{\cal K}) \label{IB_problem_a}\\
	\text{s.t.} \quad\;\;&  I(Y_{\cal T},S_{\cal T}; Z_{\cal T}| Z_{{\cal T}^c}) \leq \sum_{k \in {\cal T}} C_k, ~\forall {\cal T} \subseteq {\cal K}, \label{IB_problem_b}
	\end{align}
\end{subequations}
where we use $\Kc = \{1,2\}$ to denote the set of relay indices, and where 
$Z_{\cal K}$ is such that 
$Z_k \mk (Y_k, S_k) \mk (X, Y_{\Kc\setminus k },S_{\Kc\setminus k })$ 
forms a Markov chain  for all $k \in \Kc$. Notice that the maximization on the conditional marginal distributions
$P_{Z_k| Y_k, S_k}$ in \eqref{IB_problem} and the Markov chain condition indicates the fact that 
each relay is only aware of its own channel state. 
Intuitively, we can identify $I(X; Z_{\cal K})$ with the communication rate and 
$I(Y_{\cal T}, S_{\cal T}; Z_{\cal T}| Z_{{\cal T}^c})$ with the compression rate at the relays. 
Unfortunately, problem \eqref{IB_problem} is quite intractable. 
In the following subsections, we derive a computable 
capacity upper bound and several achievable schemes which provide computable lower bounds. 

\subsection{Upper Bounds} 
\label{ub1}

An obvious upper bound to the capacity in (\ref{IB_problem}) can be obtained by assuming that the 
channel state $\{S_k:k\in \Kc\}$ is known to both relays and the CP. 
This bound, referred to as the ``informed receiver'' upper bound and
denoted by $R^{\rm ub}_0$,  can be obtained by considering the following problem
\begin{subequations}\label{IB_problem_ub}
	\begin{align}
	\mathop {\max }\limits_{\{P_{Z_k| Y_k, S_{\Kc}}:k\in \Kc\}} \;\;\;& I(X; Z_{\cal K}| S_{\cal K}) \label{IB_problem_ub_a}\\
	\text{s.t.} \quad\;\;\;\; &  I(Y_{\cal T}; Z_{\cal T}| Z_{{\cal T}^c}, S_{\cal K}) \leq \sum_{k \in {\cal T}} C_k, ~\forall {\cal T} \subseteq {\cal K}. \label{IB_problem_ub_b}
	\end{align}
\end{subequations}
Notice that when all nodes know the full channel states, as in 
\eqref{IB_problem_ub}, then a {\em parallel channel decomposition} is possible, such that 
for each given state $\{S_k = s_k: k \in \Kc\}$ we have a model with fixed, deterministic and known-to-all states.
Then, the capacity is immediately given as the average over the channel state distribution of the 
capacities for each state, with optimization over the relay message rate allocation functions. 

The parallel channel decomposition motivates us to consider the capacity for 
fixed and deterministic channel state. 
In this case, \( S_{\cal K} \) in \eqref{IB_problem_ub} can be omitted, leading to an alternative and simplified expression for the capacity in \eqref{IB_problem_ub} as \cite{dandervoich2008communication}
\begin{align}\label{eq:equavalent_IB}
    R = \mathop {\max }\limits_{\{P_{Z_k| Y_k}:k\in \Kc\}} \left\{\mathop {\min }\limits_{ {\cal T} \subseteq {\cal K}} \left\{ I(Z_{\Tc^c}; X) + \sum_{k \in {\cal T}} \left(C_k - I(Y_k; Z_k|X)\right)\right\}\right\},
\end{align}
which is amenable for explicit computation. In particular, for the case at hand with $X \sim \Cc\Nc(0, 1)$, the optimal value of problem (\ref{eq:equavalent_IB}) is given by \cite[Theorem~$5$]{dandervoich2008communication} 
\begin{align}\label{R_fixed_rho}
 R (\rho_{\cal K}, C_{\cal K}) = \mathop {\max }\limits_{r_{\cal K} \geq 0} \left\{ \mathop {\min }\limits_{ {\cal T} \subseteq {\cal K}} \left\{ \log \left[ 1 + \sum_{k \in {\cal T}^c} \rho_k \left( 1 - 2^{-r_k} \right) \right] + \sum_{k \in {\cal T}} (C_k - r_k) \right\} \right\},
\end{align}
where for $k \in \Kc$ we define $\rho_k = |s_k|^2/\sigma^2$ to be the channel SNR, 
and $\{r_k\}$  are the optimization variables, with the meaning 
\( r_k = I(Y_k; Z_k|X) \) \cite{dandervoich2008communication}.

For later use, we rewrite \eqref{R_fixed_rho} in a convenient convex optimization form, 
that can be optimally solved using standard convex optimization tools.
Introducing an auxiliary variable $\beta$, $R (\rho_{\cal K}, C_{\cal K})$ in \eqref{R_fixed_rho} 
can be obtained as the optimal value $\beta^*$ of the problem
\begin{subequations}\label{eq_problem}
	\begin{align}
	\mathop {\max }\limits_{{r_{\cal K}}, \beta}\;\;\; & \beta \label{eq_problem_a}\\
	 \text{s.t.} \quad& \log\! \left[\! 1 \!+\! \sum_{k \in {\cal T}^c} \rho_k \left( 1 \!-\! 2^{-r_k} \right) \!\right] \!+\! \sum_{k \in {\cal T}} (C_k \!-\! r_k) \!\geq\! \beta, ~\forall {\cal T} \!\subseteq\! {\cal K}, \label{eq_problem_b}\\
	& r_k \geq 0, ~\forall k \in {\cal K}. \label{eq_problem_c}
	\end{align}
\end{subequations}
\begin{remark} \label{re: operational_meaning}
It should be noted that the achievability of \eqref{R_fixed_rho} can be established (see \cite{dandervoich2008communication}) by letting each relay \( k \) use a random codebook, generated i.i.d.  according to the auxiliary random variable \( Z_k \sim  X + \frac{N_k}{s_k} + Q_k \), where $ \frac{N_k}{s_k} \sim \Cc\Nc(0, \frac{\sigma^2}{|s_k|^2})$, and \( Q_k \sim \Cc\Nc(0,\frac{\sigma^2 2^{-r_k}}{{|s_k|^2}(1 - 2^{-r_k})}) \) for some \( r_k \geq 0 \). {In this setup, each relay generates its quantization index using Wyner-Ziv binning, based on its observation \( Y_k^n \) and its own codebook.}
The optimization problem in \eqref{R_fixed_rho} implies that the optimal codebook design for each relay depends on the SNRs of both relays.
\end{remark}

As said before, the fixed channel state case in \eqref{R_fixed_rho} can be extended to the random i.i.d. channel state 
scenario by using the parallel channel decomposition. Following this approach, 
the informed receiver upper bound $R_0^{\rm ub}$ can be obtained by solving the problem
\begin{subequations}\label{ergodic_problem}
	\begin{align}
	\mathop {\max }\limits_{\{ c_k (\rho_{\Kc}) :k\in \Kc\}} \quad & {\mathbb E} \left[ R (\rho_{\cal K}, \{ c_k (\rho_{\Kc}) \} ) \right] \label{ergodic_problem_a}\\
	\text{s.t.} \quad\; &  {\mathbb E} \left[ c_k (\rho_{{\cal K}}) \right] \leq C_k, ~\forall k \in {\cal K}, \label{ergodic_problem_b}\\
	& c_k (\rho_{{\cal K}}) \geq 0, ~\forall k \in {\cal K}, \label{ergodic_problem_c}
	\end{align}
\end{subequations}
where   $\{\rho_k = {|S_k|^2}/{\sigma^2}:k\in \Kc\}$ are   random variables  and the expectation is taken over their joint distribution.
The functions in \( \{c_k (\rho_{\Kc}):k\in \Kc \}\) specify the allocation of the corresponding total link capacities \( \{C_k \} :k\in \Kc \) 
over the parallel channels.  Although the optimal value of $R (\rho_{\cal K}, C_{\cal K})$ for each fixed $\{\rho_k :k\in \Kc \}$ in (\ref{R_fixed_rho}) can be obtained by solving its equivalent and convex transformation (\ref{eq_problem}), this does not yield a closed-form or easily computable form of the solution of  \eqref{ergodic_problem} because the optimization over the space of rate allocation functions \( \{c_k (\rho_{\Kc}) :k\in \Kc \}\) is far from obvious.  Therefore, unlike similar problems such as \cite[(6)]{caire2018information}  and \cite[(4)]{IBxu}, which admit closed-form solutions, the problem \eqref{ergodic_problem} is still difficult to handle analytically  or even numerically.

To obtain an easy-to-evaluate upper bound, we further relax the problem by assuming that the relays can cooperate. 
This turns the network into a system with a transmitter and a two-antenna oblivious relay, connected to the CP via 
a link with capacity $\sum_{k \in {\cal K}} C_k$. 
This bound, referred to as the ``cooperative informed receiver'' upper bound and denoted by $R^{\rm ub}$, 
is the solution of the problem
\begin{subequations}\label{IB_problem_ub2}
	\begin{align}
	\mathop {\max }\limits_{\left\{P_{Z_{k}| Y_{k}, S_{\Kc}} :k\in \Kc \right\}} \quad & I(X; Z_{\cal K}| S_{\cal K}) \label{IB_problem_ub2_a}\\
	\text{s.t.} \;\quad\quad\;\;\; &  I(Y_{\cal K}; Z_{\cal K}| S_{\cal K}) \leq \sum_{k \in {\cal K}} C_k. \label{IB_problem_ub2_b}
	\end{align}
\end{subequations}
Defining the channel state vector $\bm S = [S_1, S_2]^\transp \in \mathbb{C}^{2 \times 1}$ and noticing that  
the matrix $\bm S \bm S^\herm$ has rank 1 and single positive eigenvalue $\lambda$ with probability density function (pdf) \cite[(A17)]{IBxu}
\begin{equation}\label{pdf_lambda}
f_{\lambda} (\lambda) = \lambda e^{-\lambda}, ~\forall \lambda \geq 0, 
\end{equation}
using \cite[Theorem~1]{IBxu}, the optimal objective value of problem (\ref{IB_problem_ub2}) is given by
\begin{equation}\label{R_up_KM}
R^{\text {ub}} = \int_{\nu \sigma^2}^{\infty} \left[ \log \left(1 + \frac{\lambda}{\sigma^2} \right) - \log (1 + \nu)\right] f_\lambda (\lambda) d \lambda,
\end{equation}
where $\nu$ is chosen such that the following bottleneck constraint is met
\begin{equation}\label{bottle_constr_KM}
\int_{\nu \sigma^2}^{\infty} \left( \log \frac{\lambda}{\nu \sigma^2} \right) f_\lambda (\lambda) d \lambda = \sum_{k \in {\cal K}} C_k.
\end{equation}
The following lemma is a ``sanity check'', considering the limit case of $R^{\rm ub}$ when the noise power goes to zero or the  capacity constraint goes to infinity. 
\begin{lemma}[Limiting Behavior of \( R^{\rm ub} \)]\label{re:R_ub}
   Consider the cooperative informed receiver upper bound \( R^{\rm ub} \) in \eqref{R_up_KM}. 
    \begin{itemize}
        \item When the total capacity \( \sum_{k \in {\cal K}} C_k \) is fixed and SNR goes to infinity ($\sigma \rightarrow 0$), 
        \( R^{\rm ub} \) asymptotically approaches \( \sum_{k \in {\cal K}} C_k \).
        \item When SNR is fixed and \( \sum_{k \in {\cal K}} C_k \) goes to infinity, we have
   \begin{align}
       R^{\rm ub} &\rightarrow I(X; Y_{\Kc}, S_{\Kc}) \nonumber \\
       &= \int_{0}^{\infty} \left[ \log \left(1 + \frac{\lambda}{\sigma^2} \right)\right] f_\lambda (\lambda) \, d \lambda.
   \end{align}
\end{itemize}
\end{lemma}
\begin{proof}
See Appendix \ref{R_ub}.
\end{proof}

\subsection{Evaluation of the Cooperative Informed Receiver Upper Bound}
\label{baselines}

The cooperative informed receiver upper bound $R^{\rm ub}$ obtained from \eqref{IB_problem_ub2} is looser than the informed receiver upper bound $R^{\rm ub}_0$ obtained from \eqref{ergodic_problem}.
Given the difficulty of handling $R_0^{\rm ub}$ directly, in this subsection
we derive a lower bound and an approximation to $R_0^{\rm ub}$ that can be easily computed. 
It is important to note that the lower bound (as well as the approximation) to $R^{\rm ub}_0$ is not a valid capacity upper bound, 
since it is, by definition, a lower bound to an upper bound. 

\subsubsection{\bf Lower bound to $R^{\rm ub}_0$}
\label{ch:R^{ub3}}
A lower bound to $R^{\rm ub}_0$ can be obtained by setting the capacity allocation functions 
 \( c_k (\rho_{{\cal K}}) = C_k \) regardless of channel states. This choice certainly satisfies 
 the constraints \eqref{ergodic_problem_b} and \eqref{ergodic_problem_c} and therefore is feasible, but generally suboptimal. 
 The lower bound is then given by
\begin{equation}
    \check{R}^{\rm ub} = {\mathbb E} \left[ R (\rho_{\cal K}, C_{\cal K} ) \right], \label{ergodic_problem_aa}
\end{equation}
where \( R (\rho_{\cal K}, C_{\cal K} ) \) is defined in \eqref{R_fixed_rho} and can be obtained by solving the convex problem \eqref{eq_problem}.  The expectation in \eqref{ergodic_problem_aa} can be estimated by Monte Carlo (MC) simulation by sampling numerous instances of instantaneous channel SNRs \( \rho_{\cal K} \).

\subsubsection{\bf Approximation to \( R^{\rm ub}_0 \)} 
We apply the stochastic optimization framework based on the DPP approach  \cite{georgiadis2006resource, neely2010stochastic}
to handle the expectation constraints in  \eqref{ergodic_problem}. 
This approach transforms the expectation minimization into an opportunistic per-slot optimization 
problem, leveraging only current observations while avoiding long-term dependencies. 
Problem \eqref{ergodic_problem} can be reformulated as  
\begin{subequations}\label{ergodic_problem_LB} 
	\begin{align}
     \mathop{ \min}_{\{c_{k}(t, \rho_{\Kc} (t))\} :k\in \Kc, t \in \{0, 1, \cdots, T-1\} } \quad \;\;&-\lim_{T\rightarrow \infty}\frac{1}{T}\sum_{t =0}^{T-1} {\mathbb E} [ R (\rho_{\cal K}(t), \{c_{k}(t, \rho_{\Kc} (t))\} ) ] \label{ergodic_problem_a_d}\\
	\text{s.t.} \quad\quad~~~ & \lim_{T\rightarrow \infty}\frac{1}{T}\sum_{t =0}^{T-1} {\mathbb E}[c_{k}(t, \rho_{\Kc} (t))] \leq C_k, ~\forall k \in {\cal K}, \label{ergodic_problem_b_d}\\
	 &0 \leq c_k (t, \rho_{\Kc}(t)) \leq C_{\max} , ~\forall k \in {\cal K},~ t \in \{0, 1, \cdots, T-1\}, \label{ergodic_problem_c_d}
	\end{align}
\end{subequations}
where \( t \) denotes the time index, and the channel states are i.i.d. over time.
At each time step \( t \), the allocated capacity \( c_k(t, \rho_{\Kc}(t)) \) is dynamically optimized based on the current user-relay channel SNRs \( \rho_{\Kc}(t) \).  
In accordance with the standard requirements outlined in \cite{georgiadis2006resource},  the optimization variable \( c_k(t, \rho_{\Kc}(t)) \) must be bounded. To ensure this, an additional constraint \eqref{ergodic_problem_c_d} is introduced, which sets an upper bound on \( c_k(t, \rho_{\Kc}(t)) \) with a predefined parameter \( C_{\max} \geq \max\{C_k : k \in \Kc\} \). This parameter \( C_{\max} \) should be chosen large enough such that the approximation accuracy is not affected.
In particular, \( c_k(t, \rho_{\Kc}(t)) \) may occasionally exceed \( C_k \) as long as its long-term average 
expectation remains within \( C_k \). 

Let  \( V > 0 \) be a given control parameter of the DPP algorithm and let
 \( \{Q_k(t)\} \) be iteratively updated weights usually referred to as {\em virtual queue backlogs} (see
\cite{georgiadis2006resource, neely2010stochastic}), with initial condition $Q_k(t=0) =0, \forall k\in \Kc$.  
Along the ``time'' iterations $t = 0, 1, 2, \ldots$, generate a sequence of i.i.d. channel states 
$\rho_\Kc(t)$. 
The allocated `instantaneous'' capacities  \( \{c_k(t, \rho_{\Kc}(t)):k\in \Kc, t \in \{0, 1, \cdots, T-1\}\} \) at time iteration $t$ are given by the optimal values
$\{c_k^*(t):k\in \Kc, t \in \{0, 1, \cdots, T-1\}\}$  of the variables $\{c_k:k\in \Kc\}$ of the convex optimization problem:
\begin{subequations}\label{eq:solve_problem}
    	\begin{align}
		\mathop{ \min}_{ c_{\Kc}, \beta, r_{\cal K}} \quad\;\; &\sum_{k \in {\cal K}} c_k Q_k(t) -V \beta  \\
	\text{s.t.} ~\quad~~ &0 \leq c_k \leq C_{\max}, \forall k \in \Kc, \\
 & \log \!\!\left[\! 1 \!+\!\!\! \sum_{k \in {\cal T}^c} \rho_k(t) \left( 1 \!-\! 2^{-r_k} \right) \!\right] \!\!+\!\! \sum_{k \in {\cal T}} (c_k  \!-\! r_k) \!\geq\! \beta, \forall {\cal T} \!\subseteq\! {\cal K}, r_k \geq 0, ~\forall k \in {\cal K},
	\end{align}
\end{subequations}
The virtual queues evolve over time according to: 
\begin{align}
    Q_k(t+1) = \max \left[ Q_k(t) + c_k(t, \rho_{\Kc}(t)) - C_k, 0 \right], \quad \forall k \in \Kc, \;\; t = 0,1,2 \ldots \label{eq:virtual_queue}
\end{align}
Notice that the queue stability condition imposes that the average arrival rate in the $k$-th queue, given by 
$\lim_{T\rightarrow \infty}\frac{1}{T}\sum_{t =0}^{T-1} c_k(t, \rho_{\Kc}(t))$, 
must be smaller than the deterministic service rate equal to $C_k$. Hence, queue stability effectively imposes 
the average constraints \eqref{ergodic_problem_b}.  
The DPP algorithm for stochastic optimization in the presence of average constraints can be interpreted as the maximization of the time-averaged objective function along the trajectory of a dynamic system, subject to a queue-stability condition, 
where we have one virtual queue for each average constraint. 
Also, the optimal value $\beta^*(t)$ of the variable $\beta$ in  
problem~\eqref{eq:solve_problem}, yields the instantaneous 
value of the rate, whose long-term average yields the sought approximation, i.e., 
\begin{align}
    \widetilde{R}^{\rm ub} = \lim_{T \rightarrow \infty} \frac{1}{T} \sum_{t=0}^{T-1} \beta^*(t).  \label{ziobaba}
\end{align}
For given penalty parameter $V$ and finite $T$, this approximation is  a random 
variable (with some abuse of notation also denoted by $\widetilde{R}^{\rm ub}$ ), 
whose average can be evaluated by MC simulation over many runs. 
The convergence in probability of the limit in \eqref{ziobaba} is ensured by the fact that the 
channel states are i.i.d. over time and follows from the general theory developed in \cite{georgiadis2006resource, neely2010stochastic}. In other words, for sufficiently large $T$ the quantity \eqref{ziobaba} is {\em self-averaging}
and converges in probability to the limit w.r.t. of its expectation (w.r.t. the random sequence of channel states), such that MC simulation over many runs is not needed. The main steps to compute the approximation $\widetilde{R}^{\rm ub}$ are 
summarized in Algorithm \ref{alg:loop}.

\begin{algorithm}[H] 
\caption{Informed receiver upper bound approximation $\widetilde{R}^{\rm ub}$}
\label{alg:loop}
\begin{algorithmic}[1]
\State{\textbf{Initialization}: Virtual Queues $\{Q_k(0)\}$, penalty weight $V$, maximum iterations $T$, 
link capacities $\{C_k\}$}
  \For{$t = 0: T-1$}  
       \State{Generate $\{\rho_k(t) : k \in \cal{K}\}$ according to their distribution;}
        \State{Obtain optimal variable $\{c^*_k : k \in \Kc\}$  and $\beta^*(t)$ by solving \eqref{eq:solve_problem};} 
        \State {\textbf{Update} virtual queues $Q_k(t+1)$ according to \eqref{eq:virtual_queue}.}
    \EndFor
    \State {\textbf{Return}: $\widetilde{R}^{\rm ub} = \frac{1}{T} \sum_{t=0}^{T-1} \beta^*(t)$ .}
\end{algorithmic}
\end{algorithm}

The following theorem quantifies the gap between the approximation \( \widetilde{R}^{\rm ub} \) and the 
value \( R^{\rm ub}_0 \) of the original optimization problem \eqref{ergodic_problem}.  

\begin{theorem}\label{th:optimal_analysis}
Suppose the channel states \( \{S_k(t):k \in \Kc, t\in \{0,...,T-1\} \} \) are i.i.d. processes. 
Given the constants \( C_{\max} \) and penalty parameter \( V \), let  
\[
B = \frac{1}{K} \sum_{k\in \Kc} \max\{ C_k^2, (C_{\max} - C_k)^2\}.
\]  
Then, as \( T \rightarrow \infty \), the approximation satisfies  
\begin{align}\label{eq:bound_approximation}
\widetilde{R}^{\rm ub} \geq R^{\rm ub}_0  - \frac{B}{V}.
\end{align}
\end{theorem}  

\begin{proof}
    See Appendix \ref{sec:DDP_theorem_proof}.
\end{proof}

From Theorem \ref{th:optimal_analysis} we notice that 
the approximation error is inversely proportional to \( V \), ensuring that a sufficiently large \( V \) and $T$ 
can provide a very accurate approximation to the informed receiver upper bound.

\subsection{Achievable Schemes under Rayleigh Fading Channels}
\label{achiev_schems}
In this subsection, we present three achievable schemes under Rayleigh fading channels by employing explicit estimation and compression strategies at the relays. Each scheme provides a capacity lower bound and is then compared to the analytically proposed upper bound.  
\subsubsection{\bf Quantized channel inversion (QCI) scheme}
\label{QCI_scheme}

In our first scheme, each relay first gets an estimate of the channel input using channel inversion, and then transmits the quantized noise levels as well as the quantized noisy signal to the CP.
In particular, using channel inversion to $Y_k$, i.e., multiplying $Y_k$ by $\frac{S_k^*}{|S_k|^2}$, we get
\begin{align}\label{X_tilde}
{\widetilde X}_k &= X + \frac{S_k^*}{|S_k|^2} N_k \nonumber \\
&\triangleq X + \sqrt{\xi_k} N_k',
\end{align}
where $\xi_k = |S_k|^{-2}$, $N_k' = e^{-j \phi_k} N_k$, and $\phi_k$ denotes the phase of channel state $S_k$.
Due to the fact that the noise $N_k$ is rotationally invariant, $N_k'$ has the same statistics as $N_k$, i.e., $N_k' \sim {\cal CN}(0, \sigma^2)$.
We fix a finite grid of $J$ positive quantization points ${\cal B} = \{ b_1, \cdots, b_J \}$, where $b_1 \leq b_2 \leq \cdots \leq b_{J-1} < b_J$, $b_J = + \infty$, and define the following quantization operation by mapping $\xi_k$ to the closest upper quantization level:
\begin{equation}\label{ceiling_2}
\ceil[\big]{\xi_k}_{\cal B} = \min_{b \in {\cal B}} \{ \xi_k \leq b  \}.
\end{equation}
Then, each relay forces the channel (\ref{X_tilde}) to belong to a finite set of quantized noise levels by adding artificial noise, i.e., by introducing physical degradation as follows
\begin{align}\label{X_hat}
{\widehat X}_k & = {\widetilde X}_k + \sqrt{\ceil[\big]{\xi_k}_{\cal B} - \xi_k} N_k'' \nonumber\\
& = X + \sqrt{\xi_k} N_k' + \sqrt{\ceil[\big]{\xi_k}_{\cal B} - \xi_k} N_k'',
\end{align}
where $N_k'' \sim {\cal CN}(0, \sigma^2)$ is a locally generated noise independent of everything else.
Since the relay $k$ knows $\xi_k$ in each channel realization, (\ref{X_hat}) is a Gaussian channel with noise power $\ceil[\big]{\xi_k}_{\cal B} \sigma^2$.
To evaluate the achievable rate, we denote the quantized instantaneous SNR of channel (\ref{X_hat}) when $\ceil[\big]{\xi_k}_{\cal B} = b_{j_k}$ by
\begin{equation}\label{rho_hat}
{\widehat \rho}_{k, j_k} = \frac{1}{b_{j_k} \sigma^2}, ~\forall k \in {\cal K}, j_k \in {\cal J},
\end{equation}
where ${\cal J} = \{ 1, \cdots, J \}$, and define the probabilities
\begin{equation}\label{P_jk}
{\widehat P}_{k, j_k} = {\mathbb{P}} \left\{ \ceil[\big]{\xi_k}_{\cal B} = b_{j_k} \right\}, ~\forall j_k \in {\cal J}.
\end{equation}
From \cite[Theorem 5.3.1]{cover2012elements}, it is known that the minimum number of quantization bits necessary for compressing $\ceil[\big]{\xi_k}_{\cal B}$ and transmitting it to the CP is 
\begin{align}\label{H_k}
{\widehat H}_k &= H(\ceil[\big]{\xi_k}_{\cal B}) \nonumber \\
&=  - \sum_{j_k = 1}^J {\widehat P}_{k, j_k} \log {\widehat P}_{k, j_k}.
\end{align}
{To encode a sequence of $\widehat{X}_k$ of length $n$, we allocate $n \widehat{H}_k$ bits for transmitting the sequence of the quantized noise variance and the remaining $n(C_k - \widehat{H}_k)$ bits for encoding the vector quantization index of $\widehat{X}_k^n$. 
Let $\widehat{R}_{j_1, j_2}$ denote the achievable rate when $\ceil[\big]{\xi_1}_{\cal B} = b_{j_1}$ and $\ceil[\big]{\xi_2}_{\cal B} = b_{j_2}$, and let $\widehat{c}_{k, j_k}$ denote the allocated link rate for vector quantization index when $\ceil[\big]{\xi_k}_{\cal B} = b_{j_k}$. 
Since the CP receives both sequences of 
quantization noise variances from the two relays, it can group the corresponding signal samples
according to the joint (quantized) channel state $(j_1, j_2)$. Hence, 
the parallel channel decomposition over the quantized channel states can be applied to obtain the 
achievable rate for the QCI scheme, denoted as $R^{\rm qci}$. This is given as the 
optimal objective value of the problem
\begin{subequations}\label{ergodic_QCI}
	\begin{align}
	\mathop {\max }\limits_{\{\widehat{c}_{k, j_{k}}, \widehat{r}_{k, j_k}: j_k \in \mathcal{J}, k \in \cal{K}\}}   \quad &  \sum_{j_1=1}^J \sum_{j_2=1}^J \widehat{P}_{1, j_1}  \widehat{P}_{2, j_2} \widehat{R}_{j_1, j_2} \\
	\text{s.t.} \;\;\;\;\;\;\quad \quad&   \sum_{j_k =1}^J \widehat{c}_{k, j_k} \widehat{P}_{k, j_k}\leq C_k - \widehat{H}_k, ~\forall k \in {\cal K} \\
	& \widehat{c}_{k, j_k} \geq 0, ~\forall k \in {\cal K}, ~ \forall j_k \in \mathcal{J},
	\end{align}
\end{subequations}
where the optimization variables include the rate allocation $\widehat{c}_{k, j_k}$ and intermediate variables $\widehat{r}_{k, j_k}$. In particular,  $\widehat{R}_{j_1, j_2}$ is the achievable rate when $\ceil[\big]{\xi_1}_{\cal B} = b_{j_1}$ and 
$\ceil[\big]{\xi_2}_{\cal B} = b_{j_2}$, is given by 
\begin{align}\label{eq:QCI_R_j1_j2}
    \widehat{R}_{j_1, j_2} =  \mathop {\min }\limits_{ {\cal T} \subseteq {\cal K}} \left\{ \log \left[ 1 + \sum_{k \in {\cal T}^c} {\widehat \rho}_{k, j_k} \left( 1 - 2^{-\widehat{r}_{k, j_k}} \right) \right] + \sum_{k \in {\cal T}} (\widehat{c}_{k, j_k} - \widehat{r}_{k, j_k}) \right\},
\end{align} 
Notice that in this scheme each relay uses only its own channel state information knowledge, while all the values of 
the intermediate variables are determined by the quantized levels and their probabilities, which are known since the channel state statistics are known. Therefore, the scheme is compliant with the requirement of local channel state information at each relay.
Combining \eqref{ergodic_QCI} and \eqref{eq:QCI_R_j1_j2}, by incorporating auxiliary random variables $\widehat{\beta}_{j_1, j_2}$, the achievable rate for the QCI can be reformulated as the convex optimization problem
\begin{subequations}\label{ergodic_QCI_2}
	\begin{align}
	\max_{\{ \widehat{\beta}_{j_1, j_2} : j_1, j_2 \in \Jc \}, \{\widehat{c}_{k, j_{k}}, \widehat{r}_{k, j_k}:j_k \in \mathcal{J}, k \in \cal{K}\}} \quad &  \sum_{j_1=1}^J \sum_{j_2=1}^J \widehat{P}_{1, j_1}  \widehat{P}_{2, j_2} \widehat{\beta}_{j_1, j_2} \\
	\text{s.t.} \;\;\;\;\;\quad \quad \quad \quad \quad&   \sum_{j_k =1}^J \widehat{c}_{k, j_k} \widehat{P}_{k, j_k}\leq C_k - \widehat{H}_k, ~\forall k \in {\cal K}, \\
	& \widehat{c}_{k, j_k} \geq 0, ~\forall k \in {\cal K}, ~ \forall j_k \in \mathcal{J},\\
    & \log \left[ 1 + \sum_{k \in {\cal T}^c} {\widehat \rho}_{k, j_k} \left( 1 - 2^{-\widehat{r}_{k, j_k}} \right) \right] + \sum_{k \in {\cal T}} (\widehat{c}_{k, j_k} - \widehat{r}_{k, j_k}) \leq \widehat{\beta}_{j_1, j_2},  \nonumber \\
    &\quad \quad \quad \quad \quad \quad \quad \quad \quad \quad 
 \forall \Tc \subseteq{\Kc}, \forall j_k \in \Jc, \forall k \in \Kc.
	\end{align}
\end{subequations}
}

\subsubsection{\bf Truncated channel inversion (TCI) scheme}
\label{TCI_scheme}

In the second scheme, we consider a threshold $S_{\text {th}}$ on the magnitude of the state $S_k$.
{Specifically, relay $k$ selects all symbols $\widetilde{X}_k$ in \eqref{X_tilde}, which satisfy $|S_k| \geq S_{\text {th}}$, collects them into a vector, quantizes this vector and transmits the quantization index to the CP. In addition, the relay encodes and transmits a binary selection sequence to indicate whether each symbol was selected.
}

To derive an achievable rate for this scheme, we first define the probabilities of selection as
\begin{equation}\label{Prob_TCI}
{\widetilde P}_k = {\mathbb{P}} \left\{ |S_k| \geq S_{\text {th}} \right\}, ~\forall k \in {\cal K},
\end{equation}
and denote the entropy of the selection sequence as
\begin{align}\label{sigma_rho_tilde}
{\widetilde H}_k & = - {\widetilde P}_k \log {\widetilde P}_k - (1 - {\widetilde P}_k) \log (1 - {\widetilde P}_k). 
\end{align}
According to \eqref{X_tilde}, under the condition $|S_k| \geq S_{\rm th}$, the effective channel noise variance is
   \begin{align}
    {\widetilde \sigma}_k^2 & = {\mathbb E} \left[ |S_k|^{-2} \sigma^2|~ |S_k| \geq S_{\text {th}} \right],\label{eq:noise_1}
\end{align}
and the corresponding channel SNR is $ \widetilde{\rho}_k = \frac{1}{{\widetilde \sigma}_k^2}$.

{Since the CP receives the two selection sequences, it can order the channel outputs in four possible states: $(1,1)$, $(1,0)$, $(0,1)$, $(0,0)$. The achievable rates under these conditions, denoted as  $\widetilde{R}_{1, 1}, \widetilde{R}_{1, 0}, \widetilde{R}_{0, 1},$ and $\widetilde{R}_{0, 0}$, averaged over the corresponding state probabilities, yield the resulting achievable ergodic rate.
To encode an $n$-length sequence of $\widetilde{X}_k$, $n \widetilde{H}_k$ bits are allocated to transmit the binary selection sequence and the remaining $n(C_k - \widetilde{H}_k)$ bits to encode the vector quantization index of the samples $\widetilde{X}_k$ satisfying the selection condition. The average length of the sequence after selection is $n\widetilde{P}_k$, with a deviation $o(n)$ (in fact, the typical deviation is $O(\sqrt{n})$). Hence, for sufficiently large $n$, the rate allocated to the quantized selected sequence is given by 
\begin{align}
    \widetilde{c}_k = \frac{C_k - \widetilde{H}_k}{\widetilde{P}_k}.
\end{align}
Given that each relay has access only to the distribution, not the realization, of the other relay's SNR, the achievable rate for the TCI scheme, denoted as $R^{\rm tci}$, is the solution to the problem
\begin{subequations}\label{ergodic_TCI}
	\begin{align}
	\mathop {\max }\limits_{\{ \widetilde{r}_{k}:k\in \Kc\}}  \quad &  \widetilde{P}_1 \widetilde{P}_2 \widetilde{R}_{1, 1} + \widetilde{P}_1 (1 - \widehat{P}_2) \widetilde{R}_{1, 0} + (1 -\widetilde{P}_1)\widehat{P}_2 \widetilde{R}_{0, 1} + (1 - \widetilde{P}_1) ( 1-  \widehat{P}_2)\widetilde{R}_{0, 0}\\
	\text{s.t.} \;\;\; \quad&   \widetilde{c}_k = \frac{C_k - \widetilde{H}_k}{\widetilde{P}_k},  ~\forall k \in {\cal K}.  
	\end{align}
\end{subequations}
Note that the observation \( Y_k \), given that \( |S_k| \geq S_{\rm th} \), does not follow a Gaussian distribution. As established in \cite{lapidothmismatch1997}, an achievable scheme for encoding a non-Gaussian source while attaining a certain distortion level is to employ Gaussian codebooks. Therefore, the achievable rates with Gaussian codebooks for four states  are given by 
\begin{subequations}\label{eq:TCI_R}
    \begin{align}
    \widetilde{R}_{1, 1} &=  \mathop {\min }\limits_{ {\cal T} \subseteq {\cal K}} \left\{ \log \left[ 1 + \sum_{k \in {\cal T}^c} \widetilde{\rho}_{k} \left( 1 - 2^{-\widetilde{r}_{k}} \right) \right] + \sum_{k \in {\cal T}} (\widetilde{c}_{k} - \widetilde{r}_{k}) \right\}, \\
    \widetilde{R}_{1, 0} &=  \mathop {\min } \left\{ \log \left[ 1 +  \widetilde{\rho}_{1} \left( 1 - 2^{-\widetilde{r}_{1}} \right) \right], \widetilde{c}_{1} - \widetilde{r}_{1} \right\},  \\
 \widetilde{R}_{0, 1} &=  \mathop {\min } \left\{ \log \left[ 1 +  \widetilde{\rho}_{2} \left( 1 - 2^{-\widetilde{r}_{2}} \right) \right], \widetilde{c}_{2} - \widetilde{r}_{2} \right\},  \\
 \widetilde{R}_{0, 0} &=  0.
\end{align} 
\end{subequations}
According to the subscripts of the intermediate variables \( \widetilde{r}_{k} \), for each channel realization where $|S_k| \geq S_{\rm th}$, there is a corresponding intermediate variable \( \widetilde{r}_{k} \) at the relay $k$ to guide the design of codebooks.
Again, the relay processing only depends on its own channel state, while ignoring the other relay's channel state.
Combining \eqref{ergodic_TCI} and \eqref{eq:TCI_R}, by introducing the auxiliary random variables $\widetilde{\beta}_{1, 1}, \widetilde{\beta}_{1, 0}$ and $\widetilde{\beta}_{0, 1}$, the ergodic achievable rate for the TCI can be reformulated as a convex optimization problem:
\begin{subequations}\label{ergodic_TCI_2}
	\begin{align}
	\mathop {\max }\limits_{\{ \widetilde{r}_{k}:k\in \Kc\}, \widetilde{\beta}_{1, 1}, \widetilde{\beta}_{1, 0},\widetilde{\beta}_{0, 1}} \quad &  \widetilde{P}_1 \widetilde{P}_2 \widetilde{\beta}_{1, 1} + \widetilde{P}_1 (1 - \widehat{P}_2) \widetilde{\beta}_{1, 0} + (1 -\widetilde{P}_1)\widehat{P}_2 \widetilde{\beta}_{0, 1} \\
	\text{s.t.} \;\;\;\quad\quad\quad&    \widetilde{c}_k = \frac{C_k - \widetilde{H}_k}{\widetilde{P}_k},  ~\forall k \in {\cal K},\\
    & \log \left[ 1 + \sum_{k \in {\cal T}^c} \widetilde{\rho}_{k} \left( 1 - 2^{-\widetilde{r}_{k}} \right) \right] + \sum_{k \in {\cal T}} (\widetilde{c}_{k} - \widetilde{r}_{k}) \leq \widetilde{\beta}_{1, 1}, \forall \Tc \subseteq{\Kc},\\
    &\log \left[ 1 +  \widetilde{\rho}_{1} \left( 1 - 2^{-\widetilde{r}_{1}} \right) \right] \leq \widetilde{\beta}_{1, 0}, \\
    &\widetilde{c}_{1} - \widetilde{r}_{1}\leq \widetilde{\beta}_{1, 0},\\
    &\log \left[ 1 +  \widetilde{\rho}_{2} \left( 1 - 2^{-\widetilde{r}_{2}} \right) \right] \leq \widetilde{\beta}_{0, 1}, \\
    &\widetilde{c}_{2} - \widetilde{r}_{2}\leq \widetilde{\beta}_{0, 1}.
	\end{align}
\end{subequations}
}

\subsubsection{\bf MMSE-based scheme}
\label{MMSE_scheme}

Here we consider a scheme where each relay $k$ first produces the MMSE estimate of $X$ based on $(Y_k, S_k)$, and then compresses and forwards it to the CP.
In particular, given $(Y_k, S_k)$, the MMSE estimate of $X$ obtained by relay $k$ is
\begin{align}\label{x_bar_k}
{\bar X}_k &= \frac{S_k^*}{|S_k|^2 + \sigma^2} Y_k \nonumber \\
&= \frac{|S_k|^2}{|S_k|^2 + \sigma^2} X + \frac{S_k^*}{|S_k|^2 + \sigma^2} N_k.
\end{align}
Taking ${\bar X}_k$ as a new observation, we assume that relay $k$ quantizes ${\bar X}_k$ by choosing $P_{Z_k| {\bar X}_k}$ to be a conditional Gaussian distribution, i.e., 
\begin{align}\label{z_bar_k}
Z_k &= {\bar X}_k + Q_k \nonumber \\
&\triangleq U_k X + W_k,
\end{align}
where $Q_k \sim {\cal CN}(0, D_k)$ is independent of everything else, $U_k = \frac{|S_k|^2}{|S_k|^2 + \sigma^2}$, and $W_k = \frac{S_k^*}{|S_k|^2 + \sigma^2} N_k + Q_k$.
Let ${\bar X}_{k, {\rm g}}$ denote a zero-mean circularly symmetric complex Gaussian random variable with the same second moment as ${\bar X}_k$, i.e., ${\bar X}_{k, {\rm g}} \sim {\cal {CN}} \left(0, {\mathbb E} \left[ |{\bar X}_k|^2 \right] \right)$, and $Z_{k, {\rm g}} = {\bar X}_{k, {\rm g}} + Q_k$.
Using the fact that Gaussian input maximizes the mutual information of a Gaussian additive noise channel, we have the upper bound
\begin{align}\label{mutual_x_bar_z}
I({\bar X}_k; Z_k) &\leq I({\bar X}_{k, {\rm g}}; Z_{k, {\rm g}}) = \log \left( 1 + \frac{{\mathbb E} \left[ |{\bar X}_k|^2 \right]}{D_k} \right).
\end{align}
Then, the condition 
\begin{equation}\label{mutual_XkbarZk_Ck}
I({\bar X}_k; Z_k) \leq C_k,
\end{equation}
is verified by letting 
\begin{equation}\label{D}
D_k = \frac{{\mathbb E} \left[ |{\bar X}_k|^2 \right]}{2^{C_k} - 1}.
\end{equation}
Note that in the MMSE scheme, the quantization codebook generation at each relay
depends solely on its own channel statistics and each relay operation does not require the knowledge of
the other relay channel state. 

In the following, we first show that with (\ref{mutual_XkbarZk_Ck}), the bottleneck constraint of the two-relay system, i.e., 
\begin{equation}\label{bt_constr}
I({\bar X}_{\cal T}; Z_{\cal T}| Z_{{\cal T}^c}) \leq \sum_{k \in {\cal T}} C_k, ~\forall {\cal T} \subseteq {\cal K},
\end{equation}
is satisfied, and then provide a lower bound to the capacity $I(X; Z_1, Z_2)$.

Since $Z_1$ is independent of $Z_2$ given ${\bar X}_1$ and conditioning reduces differential entropy,
\begin{align}\label{mutual_X1barZ1_Z2}
I({\bar X}_1; Z_1| Z_2) & = h(Z_1| Z_2) - h(Z_1| {\bar X}_1, Z_2)\nonumber\\
& \leq h(Z_1) - h(Z_1| {\bar X}_1) \nonumber \\
&= I({\bar X}_1; Z_1) \leq C_1.
\end{align}
Analogously, we also have
\begin{align}\label{mutual_X2barZ2_Z1}
I({\bar X}_2; Z_2| Z_1) &\leq I({\bar X}_2; Z_2) \leq C_2.
\end{align}
Using the chain rule of mutual information,
\begin{align}\label{mutual_X1barX2bar_Z1Z2}
I({\bar X}_1, {\bar X}_2; Z_1, Z_2)
&= I({\bar X}_1, {\bar X}_2; Z_1) + I({\bar X}_1, {\bar X}_2; Z_2| Z_1)\nonumber\\
&= I({\bar X}_1; Z_1) \!+\! I({\bar X}_2; Z_1| {\bar X}_1) \!+\! I({\bar X}_1; Z_2| Z_1) + I({\bar X}_2; Z_2| {\bar X}_1, Z_1) \nonumber\\
&= I({\bar X}_1; Z_1) + I({\bar X}_2; Z_2| {\bar X}_1, Z_1),
\end{align}
where we used
\begin{align}\label{mutual_2}
I({\bar X}_2; Z_1| {\bar X}_1) & = I({\bar X}_2; Q_1) = 0,\\
I({\bar X}_1; Z_2| Z_1) & = I(Q_1; Z_2| Z_1) = 0.
\end{align}
since $Q_1$ is independent of everything else.
Moreover,
\begin{align}\label{mutual_X2barZ2_X1barZ1}
I({\bar X}_2; Z_2| {\bar X}_1, Z_1) & = I({\bar X}_2; Z_2| {\bar X}_1)\nonumber\\
& = h(Z_2| {\bar X}_1) - h(Z_2| {\bar X}_1, {\bar X}_2)\nonumber\\
& = h(Z_2| {\bar X}_1) - h(Z_2| {\bar X}_2)\nonumber\\
& \leq h(Z_2) - h(Z_2| {\bar X}_2)\nonumber\\
& = I({\bar X}_2; Z_2).
\end{align}
Combining (\ref{mutual_XkbarZk_Ck}), (\ref{mutual_X1barX2bar_Z1Z2}), and (\ref{mutual_X2barZ2_X1barZ1}), we have
\begin{align}
I({\bar X}_1, {\bar X}_2; Z_1, Z_2) &\leq I({\bar X}_1; Z_1) + I({\bar X}_2; Z_2) \leq C_1 + C_2.\label{mutual_X1barX2bar_Z1Z2_2}
\end{align}
From (\ref{mutual_X1barZ1_Z2}), (\ref{mutual_X2barZ2_Z1}), and (\ref{mutual_X1barX2bar_Z1Z2_2}), it follows
that the bottleneck constraint (\ref{bt_constr}) is satisfied.

The next step is to evaluate $I(X; Z_1, Z_2)$.
\begin{align}\label{mutual_XZ1Z2}
I(X; Z_1, Z_2) &= h(Z_1, Z_2) - h(Z_1, Z_2| X)\nonumber\\
&\geq  h(Z_1, Z_2| S_1, S_2) - h(Z_1, Z_2| X)\nonumber\\
&=  h(Z_1, Z_2| S_1, S_2) - h(Z_1| X) - h(Z_2| X),
\end{align}
where the last step holds because $Z_2 \mk X \mk Z_1$.
Then, we evaluate the terms in (\ref{mutual_XZ1Z2}) separately. Let
\begin{align}\label{var_W1}
V_k \triangleq {\text {Var}} (W_k| S_k) &= \frac{U_k \sigma^2}{|S_k|^2 + \sigma^2} + D_k.
\end{align}
Since $X$, $N_k$, and $Q_k$ are independent normal variables, given $(S_1, S_2)$, $Z_1$ and $Z_2$ are jointly Gaussian.
Hence,
\begin{align}\label{mutual_XZ1Z2_S1S2}
h(Z_1, Z_2| S_1, S_2)
=  {\mathbb E} \left[ \log (\pi e)^2 \det \left( \begin{bmatrix} U_1^2 & U_1 U_2\\ U_1 U_2 & U_2^2 \end{bmatrix}
+ \begin{bmatrix} V_1 & 0 \\ 0 & V_2 \end{bmatrix}
\right) \right].
\end{align}
Moreover, since the Gaussian distribution maximizes the entropy over all distributions with the same variance \cite{el2011network}, we have
\begin{equation}\label{zk_X}
h(Z_k| X) \leq {\mathbb E} \left[ \log \pi e \left( {\text {Var}} (U_k) |X|^2 + {\mathbb E} \left[ V_k \right] \right) \right].
\end{equation}
Substituting (\ref{mutual_XZ1Z2_S1S2}) and (\ref{zk_X}) into (\ref{mutual_XZ1Z2}), 
the achievable rate for this scheme is given by 
\begin{align}\label{R_lb_MMSE}
R^{\text {mmse}} & = {\mathbb E} \left[ \log \det \left( \begin{bmatrix} U_1^2 & U_1 U_2\\ U_1 U_2 & U_2^2 \end{bmatrix}
+ \begin{bmatrix} V_1 & 0 \\ 0 & V_2 \end{bmatrix}
\right) \right] - \sum_{k=1}^K {\mathbb E} \left[ \log \left( {\text {Var}} (U_k) |X|^2 + {\mathbb E} \left[ V_k \right] \right) \right].
\end{align}

\section{The Two-user Case}
\label{sec:two_user_IB}

In this section, we consider the disaggregated RAN model for the two-user case. 
To evaluate the sum capacity, we also propose an analytical upper bound. 
As we will explain later, with two users, the channel inversion for single-antenna relays is not possible (the $1 \times 2$ channel vector at each relay is a non-invertible matrix) and MMSE estimation of the transmitted signal vector from each individual observation (scalar) at relays does not provide any effective local denoising.  Hence, the achievable schemes developed for the single-user case are not useful for the two-user case. As a result, we propose an achievable scheme 
based on joint quantization of both the channel state and the received signal at each relay. 
Of course, this achievable scheme combined with the upper bound traps the optimal sum capacity, which
for this model appears to be a very challenging open problem.

\begin{figure}[ht!]
\centering
\begin{tikzpicture}[node distance = 0.03\textwidth]
\tikzstyle{neuron} = [circle, draw=black, fill=white, minimum height=0.05\textwidth, inner sep=0pt]
\tikzstyle{rect} = [rectangle, rounded corners, minimum width=0.05\textwidth, minimum height=0.05\textwidth,text centered, draw=black, fill=white]
    \node [neuron] (neuron1) {\small{User 1}};
    \node [left of=neuron1, xshift=-0.03\textwidth] (X1) {$X_1$};
    \node [below of=neuron1, yshift=-0.1\textwidth, neuron] (neuron2) {\small{User 2}};
     \node [left of=neuron2, xshift=-0.03\textwidth] (X2) {$X_2$};
     \node [right of=neuron1, xshift=0.2\textwidth, neuron] (add1) {$+$};
      \node [above of=add1, yshift=0.05\textwidth] (N1) {$N_1$};
      \node [right of=neuron2, xshift=0.2\textwidth, neuron] (add2) {$+$};
      \draw [->,line width=1pt] (neuron1) -- node[near start, above] {$h_{2, 1}$}(add2);
        \draw [->,line width=1pt] (neuron2) -- node[above] {$h_{2, 2}$}(add2);
     \draw [->,line width=1pt] (neuron1) -- node[above] {$h_{1, 1}$}(add1);
          \draw [->,line width=1pt] (neuron2) -- node[near start, above] {$h_{1, 2}$}(add1);
         \node [below of=add2, yshift=-0.05\textwidth] (N2) {$N_2$};
    \node [right of=add1, xshift=0.15\textwidth, rect] (R1) {\small{Relay 1}};
     \node [right of=add2, xshift=0.15\textwidth, rect] (R2) {\small{Relay 2}};
        \node [below right of=neuron1, xshift=0.6\textwidth, yshift=-0.05\textwidth, rect] (D) {{CP}};
     \draw [->,line width=1pt] (N1) -- (add1);
     \draw [->,line width=1pt] (add1) -- node[above] {$Y_1$}(R1);
       \draw [->,line width=1pt] (R1) -- node[above] {$W_1$}(D);
       \draw [->,line width=1pt] (N2) -- (add2);
         \draw [->,line width=1pt] (add2) -- node[above] {$Y_2$}(R2);
         \draw [->,line width=1pt] (R2) -- node[above] {$W_2$}(D);
\end{tikzpicture}
\caption{The disaggregated RAN model with two users and two relays.}  
\label{Block_diagram_2}
\end{figure}
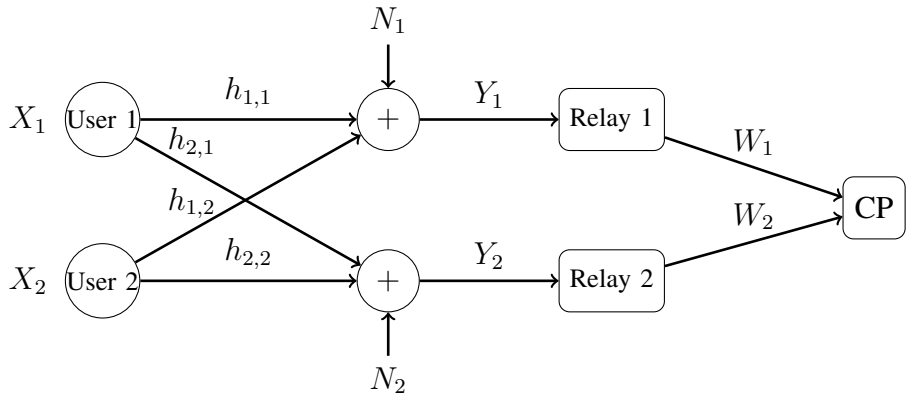

The two users model is shown in Fig. \ref{Block_diagram_2}. As for the single-user case, 
the transmitted codeword \( X_1^n \) and \( X_2^n \) are generated independently in an i.i.d. fashion 
according to a $\Cc\Nc(0,1)$ distribution. The channel transition probability is defined by the single-letter 
input-output linear Gaussian channel
\begin{align}
    Y_k = \Hm_k^\herm \Xm + N_k,
\end{align}
where $\Xm = [X_1, X_2]^{T} \sim {\cal CN}(\mathbf{0}, \Id_2)$, $\Hm_k = [h_{k, 1}, h_{k, 2}]^{\herm} \sim {\cal CN}(\mathbf{0}, \Id_2)$, and $N_k \sim {\cal CN}(0, \sigma^2)$ are respectively the channel input, channel state from the users to relay $k$, and Gaussian noise at relay $k$. 

We observe that \cite[Theorem 1]{aguerri2019TIT} derived the sum capacity for the 
general $L$-user, $K$-relay case discrete memoryless channel with oblivious relay processing
under the condition that the received signal at relay $k$ is conditionally independent of the received signal
at the other relays $\Kc \setminus \{k\}$ given the transmitted signals $X_1, \ldots, X_L$. In our case, this 
corresponds to the Markov chain condition
\[ (Y_k, \Hm_k) \mk (X_1,X_2) \mk (Y_{\bar{k}}, \Hm_{\bar{k}}), \]
that holds under the assumptions (made in this paper) that the additive noise and the channel states at the two relays are mutually independent. Hence, \cite[Theorem 1]{aguerri2019TIT} applies directly by identifying 
the ``augmented'' relay received signal as $(Y_k, \Hm_k)$, and yields
\begin{align}\label{eq:sum_rate_JD}
    R = \max_{\{P_{Z_k| Y_k, \Hm_{k}}:k\in \Kc\}} \left\{ \min_{\Tc \subseteq \Kc}~ \left\{ \sum_{k\in \Tc}~ \left[C_k - I(Y_k, \Hm_k; Z_k| \Xm)\right] + I(\Xm; Z_{\Tc^c
    } )\right\} \right\}, 
\end{align}
where $\{Z_k \sim P_{Z_k| Y_k, \Hm_{k}} : k\in \Kc\}$ are auxiliary random variables.
In addition, \cite[Theorem 4]{aguerri2019TIT} derived an \emph{equivalent} sum capacity expression in the form
\begin{subequations}\label{eq:sum_rate_SD}
	\begin{align}
	\mathop {\max }_{\{P_{Z_k| Y_k, \Hm_{k}}\}}  \quad &I(\Xm ; Z_{\Kc}) \\
	\text{s.t.}~ \quad\quad &I(Y_{\Tc}, \Hm_{\Tc}; Z_{\Tc}|Z_{\Tc^c}) \leq \sum_{k \in \Tc} C_k , ~\forall \Tc \subseteq \Kc.
	\end{align}
\end{subequations}
Unfortunately, the optimization with respect to the 
marginal conditional distributions $\{P_{Z_k| Y_k, \Hm_{k}} :k\in \Kc\}$ 
in either  \eqref{eq:sum_rate_JD} or \eqref{eq:sum_rate_SD} is far from trivial, and since the channel at hand involves continuous random variables, even a brute-force solution of these problems through numerical methods is unclear. 
Hence, in the following subsections, we propose an upper bound using \eqref{eq:sum_rate_SD} 
and an achievable scheme using \eqref{eq:sum_rate_JD}.

\subsection{Upper Bounds}
\label{informed_ub}
{
Assuming that the channel state is known to all relays and at the CP, the 
informed receiver upper bound, denoted as $R^{\rm ub}$ can be obtained by solving the following problem
\begin{subequations}\label{eq:sum_rate_informed_UB}
	\begin{align}
	\mathop {\max }_{\{P_{Z_k| Y_k, \Hm_\Kc} :k\in \Kc\}}  \quad &I(\Xm ; Z_{\Kc}|\Hm_{\cal K}) \\
	\text{s.t.} ~\quad\quad &I(Y_{\Tc}; Z_{\Tc}| \Hm_{\Tc}, Z_{\Tc^c}) \leq \sum_{k \in \Tc} C_k , ~\forall \Tc \subseteq \Kc.
	\end{align}
\end{subequations}
Notice that the optimization here is w.r.t. the conditional distributions $\{P_{Z_k| Y_k, \Hm_\Kc} \}$ where the conditioning is 
over the whole channel state $\Hm_\Kc$.

A special case of \eqref{eq:sum_rate_informed_UB}, where the channel states $\Hm_k = \hv_k, ~\forall k \in {\cal K}$ 
are fixed constants under the assumption of Gaussian inputs and Gaussian channels, has been studied in 
\cite[Theorem 5]{aguerri2019TIT}. In this case, the sum rate in 
\eqref{eq:sum_rate_informed_UB} for $\Hm_k = \hv_k, ~\forall k \in {\cal K}$ deterministic and known takes on the form
\begin{align}\label{R_fixed_rho_two_user}
& R (\hv_{\cal K}, C_{\cal K}) = \mathop {\max }\limits_{r_{\cal K} \geq 0} \left\{ \mathop {\min }\limits_{ {\cal T} \subseteq {\cal K}} \left\{\sum_{k \in {\cal T}} \left[C_k \!-\! r_k\right]+ \log \det \left[ \Id_2\! +\!\sum_{k \in {\cal{T}}^c } \frac{\left(1 \!-\! 2^{-r_k}\right)}{\sigma^2} \hv_k \hv_k^{\herm}  \right]  \right\}\right\},
\end{align}
where $r_{\cal K}$ are non-negative optimization random variables representing $\{I(Y_k; \Zm_k|\Xm)\}$
\cite{dandervoich2008communication}. 
Introducing an auxiliary optimization variable $\beta$, $R (\hv_{\cal K}, C_{\cal K})$ can be computed 
as the optimal value of the convex problem:
\begin{subequations}\label{eq_problem_two_user}
	\begin{align}
	\mathop {\max }\limits_{r_{\cal K}, \beta}\;\;\; & \beta \label{eq_problem_a_two_user}\\
	 \text{s.t.} \quad&  \sum_{k \in {\cal T}} (C_k \!-\! r_k)  \!+\! \log \det \left[ \Id_2 + \sum_{k \in {\cal{T}}^c } \frac{\left(1 - 2^{-r_k}\right)}{\sigma^2} \hv_k \hv_k^{\herm}  \right] \!\geq\! \beta, \forall~ {\cal T} \!\subseteq\! {\cal K}, \label{eq_problem_b_two_user}\\
	& r_k \geq 0, ~\forall k \in {\cal K}. \label{eq_problem_c_two_user}
	\end{align}
\end{subequations}

\begin{remark} \label{re: operational_meaning_two_user}
Similar to Remark \ref{re: operational_meaning} for the single-user case, the achievability of \eqref{R_fixed_rho_two_user} in the two-user scenario can be established by a random quantization coding ensemble at relay $k$, generated in an i.i.d. manner according to the auxiliary random variable \( Z_k \sim Y_k + Q_k = \hv_k^\herm \Xm + N_k + Q_k \), where \( Q_k \sim \Cc\Nc\left(0, \frac{2^{-r_k} \sigma^2}{\hv_k^\herm \hv_k(1 - 2^{-r_k})}\right) \) for some \( r_k \geq 0 \). Relay $k$ then creates its quantization index using Wyner-Ziv binning, based only on its observation \( Y_k^n \) and its own codebook.
The optimization problem in \eqref{eq_problem_two_user} thus implies that the optimal codebook design at each relay depends on the channel state of both relays through intermediate variables $r_k$. 
\end{remark}

The fixed channel state scenario in \eqref{R_fixed_rho_two_user} can be generalized to the random ergodic channel state scenario through parallel channel decomposition, where each realization of the channel state  $\Hm_k$ can be interpreted as a distinct ``parallel channel" occurring with probability determined by the underlying channel state distribution. 
This yields that the informed receiver upper bound \( R^{\rm ub}_0 \) in \eqref{eq:sum_rate_informed_UB} can be obtained by solving the problem
\begin{subequations}\label{ergodic_problem_two_users}
	\begin{align}
	\mathop {\max }\limits_{\{c_{k}(\Hm_{\Kc}):k\in \Kc\}} \quad & {\mathbb E} \left[ R (\Hm_{\cal K}, \{c_{k}(\Hm_{\Kc})\} ) \right] \label{ergodic_problem_two_user_a}\\
	\text{s.t.} \quad\;\quad &  {\mathbb E} \left[ c_k (\Hm_{\Kc}) \right] \leq C_k, \label{ergodic_problem_two_user_b}\\
	& c_k (\Hm_{\Kc}) \geq 0, ~\forall k \in {\cal K}, \label{ergodic_problem_two_user_c}
	\end{align}
\end{subequations}
where $R (\Hm_{\cal K}, c_{\cal K} ) $ is defined in \eqref{R_fixed_rho_two_user} and $c_k(\Hm_{\Kc})$, a function of all channel states, is the partial capacity allocated to the $k$-th relay. 
Similar to the challenges discussed in \eqref{ergodic_problem}, solving \eqref{ergodic_problem_two_users} remains an open question, although an asymptotically tight approximation can be obtained using the DPP approach as done for the single-user 
case (details are omitted for brevity). 

To derive an upper bound in analytical form, we further relax the problem \eqref{eq:sum_rate_informed_UB} by assuming that the relays can cooperate. This leads to the looser ``cooperative informed receiver'' upper bound, denoted as \( R^{\rm ub} \), which can be obtained as the solution of
\begin{subequations}\label{DIB_problem_ub2}
	\begin{align}
	\mathop {\max }\limits_{\{P_{\Zm_k| Y_k, \Hm_{\Kc}}:k\in \Kc\}} & I(\Xm; Z_{\cal K}| \Hm_{\cal K}) \label{DIB_problem_ub2_a}\\
	\text{s.t.} ~\quad\;\;\; &  I(Y_{\cal K}; Z_{\cal K}|  \Hm_{\cal K}) \leq \sum_{k \in {\cal K}} C_k. \label{DIB_problem_ub2_b}
	\end{align}
\end{subequations}

Let $\Hm = [\Hm_1, \Hm_2] \in \mathbb{C}^{2 \times 2}$.{ Since $\Hm_k \sim \Nc(\mathbf{0}, \Id_2)$, $\Hm \Hm^\herm$ is a Wishart matrix and has $T = 2$ positive eigenvalues according to the Marchenko-Pastur law in \cite{tulino2004random}.}
It is known from \cite[(A17)]{IBxu} that the probability density function (pdf) of the unordered eigenvalue $\lambda$ 
 of $\Hm \Hm^\herm$ is
\begin{equation}\label{pdf_lambda_2}
f_{\lambda} (\lambda) = \frac{1}{T} \sum_{i=0}^{T-1} \frac{i!}{(i + n)!} [L_{i}^{n} (\lambda)]^2 \lambda^{n} e^{(-\lambda)},
\end{equation}
where $n$ denotes the difference between maximum dimension and minimum dimension of $\Hm$, i.e., $n=0$ in this case, and $L_{i}^{n}$ defined below is the associated Laguerre polynomial of order $i$. 
\begin{align}
L_{i}^{n} (\lambda) = \frac{e^{\lambda}}{i! \lambda^{n}}  \frac{d^i}{d \lambda^i} \left(e^{-\lambda} \lambda^{n + i} \right).
\end{align}
{Then, following \cite[Theorem 1]{IBxu},  a closed-form solution to problem \eqref{DIB_problem_ub2} is given by}
\begin{equation}\label{R_up_KM_2}
R^{\text {ub}} = T \int_{\nu \sigma^2}^{\infty} \left[ \log \left(1 + \frac{\lambda}{\sigma^2} \right) - \log (1 + \nu)\right] f_\lambda (\lambda) d \lambda,
\end{equation}
where $\nu$ is chosen such that the following bottleneck constraint is met
\begin{equation}\label{bottle_constr_KM_2}
\int_{\nu \sigma^2}^{\infty} \left( \log \frac{\lambda}{\nu \sigma^2} \right) f_\lambda (\lambda) d \lambda = \frac{C_1 + C_2}{T}.
\end{equation}
In the following lemma, we consider the limit case of $R^{\rm ub}$.
\begin{lemma}[Limiting Behavior of \( R^{\rm ub}\)]\label{R_ub_two_user}
   Consider the cooperative informed receiver upper bound $R^{\rm ub}$ in \eqref{R_up_KM_2}.
   \begin{itemize}
       \item Given the relay-CP link capacity constraints $C_1$ and $C_2$, when SNR goes to infinity, i.e., $\sigma \rightarrow 0$, the upper bound $R^{\rm ub}$ tends asymptotically to $\sum_{k \in {\cal K}} C_k$.
       \item For given SNR, when $\sum_{k \in {\cal K}} C_k$ approaches infinity, we have 
   \begin{align}\label{eq:limit_case_C_inf}
       R^{\rm ub} &\rightarrow I(\Xm; \Ym_{\Kc}, \Hm_{\Kc}) \nonumber \\
       &= T \int_{0}^{\infty} \left[ \log \left(1 + \frac{\lambda}{\sigma^2} \right)\right] f_\lambda (\lambda) d \lambda.
   \end{align}
   \end{itemize} 
\end{lemma}
\begin{proof}
       Lemma \ref{R_ub_two_user} can be proven similarly as Lemma \ref{re:R_ub}. The proof is thus omitted for brevity.
\end{proof}



{As done in the single-user case of the previous section, we wish to trap $R_{0}^{\rm ub}$ in \eqref{ergodic_problem_two_users} between the easily computable looser upper bound $R^{\rm ub}$ in \eqref{R_up_KM_2} and a lower bound $\check{R}^{\rm ub}$. Notice again that the latter is a lower bound to an upper bound, and therefore it is still generally not an achievable rate.}

A lower bound to $\check{R}^{\rm ub} \leq R^{\rm ub}_0$ can be obtained by setting the allocated capacity for user \( k \) to a fixed value \( C_k \), regardless of the channel state. This yields
\begin{equation}
    \check{R}^{\rm ub} = \mathbb{E} \left[ R (\Hm_{\cal K}, C_{\cal K} ) \right], \label{ergodic_problem_aa_two_user}
\end{equation}
where \( R (\Hm_{\cal K}, C_{\cal K} ) \) is defined in \eqref{R_fixed_rho_two_user}.
For each channel realization \( \Hm_{\cal K} \), the rate \( R(\Hm_{\cal K}, C_{\cal K}) \) can be reformulated as a convex optimization problem \eqref{eq_problem_two_user} and computed by standard convex optimization tools. 
The expectation is computed by MC  simulation. 

\subsection{Achievable Scheme under Random i.i.d. Channel State}
\label{achive_shmes_2}

Compared to \cite{aguerri2019TIT}, a key challenge in our case is that the CP does not have direct access to the channel state, while the relays do.
To address this, we propose a method where the relays not only forward the received signals, but also transmit the channel state to the CP. We call this approach {\em fronthaul compression} (FC). In this scheme, a portion of the relay-CP link capacity is dedicated to the lossy compression of the channel state $\Hm_k$, while the remaining capacity is allocated to the transmission of the received signal $Y_k$. This approach strikes a balance between transmitting the necessary channel information and the received signal data.

In this scheme, the random quantization codebook at relay \( k \) is generated i.i.d. according to two random variables: \( \Zm_k^{(1)} \), which is used to generate the quantization codebook for the channel information \( \Hm_k \), and \( Z_k^{(2)} \), which is used to generate the quantization codebook for the received signal \( Y_k \), {given $\Zm_k^{(1)}$}. 
{The chosen underlying joint distribution of ${\Xm}, Y_k, \Hm_k, \Zm_k^{(1)}$, and $Z_k^{(2)}$ with respect to which
all the following mutual information expressions are defined is given by 
\begin{align}\label{eq:joint_pdf}
    P_{{\Xm}, Y_k, \Hm_k, \Zm_k^{(1)}, Z_k^{(2)}} = 
    P_{\Xm} P_{\Hm_k} P_{Y_k|{\Xm}, \Hm_k} P_{\Zm_k^{(1)}|\Hm_k} {P_{Z_k^{(2)}| Y_k, \Zm_k^{(1)}}}, 
\end{align}}
{where $P_{\Xm}, P_{\Hm_k}, P_{Y_k|{\Xm}, \Hm_k}$ are fixed by the input, state, and channel models, 
while $P_{\Zm_k^{(1)}|\Hm_k}$ and $P_{Z_k^{(2)}| Y_k, \Zm_k^{(1)}}$ define the auxiliary variables.} 

First, we consider encoding the channel states. The minimum number of bits per symbol (i.e., per $2 \times 1$ vector)
required to encode a sequence of channel states 
$\Hm_k^n$ with the quadratic distortion metric and given the average distortion per component $D$ 
is given by \cite{cover2012elements}
\begin{align}\label{eq:Zm_k1}
    R_k(D) &= \min_{P_{\Zm_k^{(1)}|\Hm_k}: \mathbb{E}\left[d\left(\Hm_k,\Zm_k^{(1)}\right)\right] \leq 2D} I(\Hm_k; \Zm_k^{(1)}).
\end{align}
According to \cite[Theorem 10.3.3]{cover2012elements}, \eqref{eq:Zm_k1} under Gaussian inputs can be derived as 
\begin{align}\label{eq:R_D}
    R_k(D) = 2 \log \frac{1}{D}, ~\text{where} ~ 0 \leq D \leq 1.
\end{align}
Next, we consider encoding the received signal $Y_k$. 
Let $\Em_k$ denote the error vector resulting from the quantization of $\Hm_k$, i.e., $\Em_k = \Hm_k- \Zm_k^{(1)}$. Based on \cite[Theorem 10.3.3]{cover2012elements} for Gaussian inputs, the error $\Em_k$ and the representation ${\Zm_k^{(1)}}$ both follow complex Gaussian distributions, specifically $\Em_k \sim \Cc\Nc (\mathbf{0}, D \Id)$,  and ${\Zm_k^{(1)}} \sim \Cc\Nc (\mathbf{0}, (1-D) \Id)$. Therefore, the received signal at relay $k$ can be rewritten as
\begin{align}\label{eq:y_k}
    Y_k = {\Zm_k^{(1)}}^\herm \Xm + \Em_k^\herm \Xm + N_k,
\end{align}
where the second moment of $Y_k$ given $ \Zm_k^{(1)}$ is $\mathbb{E}\left[|Y_k|^2 | \Zm_k^{(1)}\right] =  
{ \| \Zm_k^{(1)} \|^2} + 2D + \sigma^2 $. 

Replacing the general $\Zm_k$ with $(\Zm_k^{(1)},Z_2)$ in 
\eqref{eq:sum_rate_JD} the mutual information $I(Y_k, \Hm_k; \Zm_k| \Xm)$ can be written as
\begin{subequations}
    \begin{align}
    &I(Y_k, \Hm_k; \Zm_k^{(1)}, Z_k^{(2)}| \Xm) \nonumber \\
    = &I(Y_k, \Hm_k; \Zm_k^{(1)}| \Xm) + I(Y_k, \Hm_k; Z_k^{(2)}|\Zm_k^{(1)}, \Xm) \\
    = &I(\Hm_k; \Zm_k^{(1)}| \Xm) \!+\! I(Y_k;\Zm_k^{(1)}| \Hm_k, \Xm) \!+\! I(Y_k; Z_k^{(2)}|\Zm_k^{(1)}, \Xm) \! +\! I(\Hm_k; Z_k^{(2)}|Y_k, \Zm_k^{(1)}, \Xm) \label{eq:chain_1}\\
    = &I(\Hm_k; \Zm_k^{(1)}) +  I(Y_k; Z_k^{(2)}|\Zm_k^{(1)}, \Xm) \label{eq:chain_2} \\
    = &R_k(D) + I(Y_k; Z_k^{(2)}|\Zm_k^{(1)}, \Xm),\label{eq:chain_3}
\end{align}
\end{subequations}
\eqref{eq:chain_1} follows by using the chain rule of mutual information, 
\eqref{eq:chain_2} follows from the fact that, 
due to the joint distribution in~\eqref{eq:joint_pdf}, 
$I(Y_k;\Zm_k^{(1)}| \Hm_k, \Xm)$ and $ I(\Hm_k; Z_k^{(2)}|Y_k, \Zm_k^{(1)}, \Xm)$ are both 
equal to zero, and  \eqref{eq:chain_3}  follows from  \eqref{eq:Zm_k1}.
On the other hand, $I(\Xm; \Zm_{\Tc^c})$ in \eqref{eq:sum_rate_JD} can be expressed as
\begin{subequations}
    \begin{align}
    I(\Xm; Z_{\Tc^c}^{(2)}, \Zm_{\Tc^c}^{(1)}) &= I(\Xm; Z_{\Tc^c}^{(2)}| \Zm_{\Tc^c}^{(1)}) + I(\Xm;\Zm_{\Tc^c}^{(1)} ) \\
    &= I(\Xm; Z_{\Tc^c}^{(2)}| \Zm_{\Tc^c}^{(1)}), \label{eq: 2_term}
\end{align}
\end{subequations}
{where \eqref{eq: 2_term} holds, since $\Xm$ and $\Zm_{\Tc^c}^{(1)}$ are independent according to the joint distribution in \eqref{eq:joint_pdf}.} Then, by substituting \eqref{eq:chain_3} and \eqref{eq: 2_term} into \eqref{eq:sum_rate_JD}, the achievable rate obtained by the FC scheme takes on the form
\begin{align}\label{eq:sum_rate_JD_2}
    R = \max_{\left\{P_{Z_k^{(2)}| Y_k, \Hm_k}:k\in \Kc\right\}} \left\{ \min_{\Tc \subseteq \Kc} \left\{ \sum_{k \in \Tc} \left[C_k \!-\! R_k(D) \!-\! I(Y_k; Z_k^{(2)}|\Zm_k^{(1)}, \Xm)\right] \!+\! I(\Xm; Z_{\Tc^c}^{(2)}| \Zm_{\Tc^c}^{(1)}) \!\!\right\}\!\!\right\}.
\end{align}
Note that although the quantized channel coefficient $\Zm_k^{(1)}$ is known, the received signal $Y_k$ is not conditionally Gaussian given $X$ and $\Zm_k^{(1)}$ due to the term $\Em_k$ in \eqref{eq:y_k}.
Therefore,  \eqref{eq:sum_rate_JD_2} it is still analytically intractable.
 To overcome this, according to \cite{lapidothmismatch1997}, an achievable scheme for encoding a non-Gaussian source while attaining a certain distortion level is to employ Gaussian codebooks.
Therefore, after receiving non-Gaussian observation symbols $Y_k^n$, relay $k$ makes use of a quantization codebook generated in an i.i.d. fashion according to an auxiliary conditionally Gaussian random variable 
\begin{align}\label{eq:Z_k_g}
    {Z_{k, {\rm g}}^{(2)}} &= {\Zm_k^{(1)}}^\herm \Xm +  \widehat{N}_{k, {\rm g}} + \widehat{Q}_k,
\end{align}
given $\Zm_k^{(1)}$, where $\widehat{N}_{k, {\rm g}}$ is an independent 
Gaussian random variable with zero mean and the same second moment as  
$\Em_k^\herm \Xm + N_k$ in (\ref{eq:y_k}), i.e., $ \widehat{N}_{k, {\rm g}} \sim {\cal CN}(0, 2D + \sigma^2)$. 
The corresponding SNR in this user-relay channel is denoted as 
$\hat{\rho}_{k} = \frac{{\Zm_k^{(1)}}^\herm {\Zm_k^{(1)}}}{2D + \sigma^2}$.
The additive variable $\widehat{Q}_k$ represents a Gaussian random variable with variance $\frac{2^{-\hat{r}_k}}{\hat{\rho}_k(1 - 2^{-\hat{r}_k})}$, and the design of $\hat{r}_k$ will be explained later. Therefore, the achievable rate of the FC scheme is given by
\begin{align}\label{eq:sum_rate_JD_2_LB}
    {R}^{\rm fc} \!&=\! \min_{\Tc \subseteq \Kc} \!\left \{\sum_{k \in \Tc} \left[C_k - R_k(D) \!-\! I(Y_{k, {\rm g}}; Z_{k, {\rm g}}^{(2)}|\Zm_k^{(1)}, \Xm)\right] +  I(\Xm; Z_{\Tc^c, {\rm g}}^{(2)}| \Zm_{\Tc^c}^{(1)}) \right\}.
\end{align}
which can be further maximized over the variables $\{\hat{r}_k\}$. 

Suppose that the quantized channel states $\Zm_{\Kc}^{(1)}$ are constant and known at \emph{both} the relays and the CP, which has also been considered in \eqref{R_fixed_rho_two_user}. Then, the optimal value of \eqref{eq:sum_rate_JD_2_LB} can be obtained by solving 
\begin{align}\label{eq:R_ZK}
& R^{\rm fc} (\Zm_{\cal K}^{(1)}, C_{\cal K}) \!=\! \mathop {\max }\limits_{r_{\cal K} \geq 0} \left\{ \mathop {\min }\limits_{ {\cal T} \subseteq {\cal K}} \left\{\sum_{k \in {\cal T}} \left[C_k - R_k(D)- r_k\right]\!+\! \log \det\left[ \Id_2 \!+\! \sum_{k \in {\cal{T}}^c } \frac{1 \!-\! 2^{-r_k}}{2D \!+\! \sigma^2}\Zm_k^{(1)} {\Zm_k^{(1)}}^{\herm}  \right]  \right\}\right\},
\end{align}
which is analogous to  \eqref{R_fixed_rho_two_user} using the induced user-relay channel in \eqref{eq:Z_k_g}.

However, for random ergodic channel state, the parallel channel decomposition and the solution of 
\eqref{eq:R_ZK} requires that both relays know the whole quantized state $\Zm^{(1)}_\Kc$. 
This condition is not satisfied in our setting where each relay only has access to its own channel state. 
To overcome this problem, here we use a strategy referred to as the ``single-relay'' assumption, such that 
each relay determines its own variable $\hat{r}_k$ assuming that it is the only relay connected to the CP (admittedly, a very conservative assumption). It follows that the achievable rate under the single-relay assumption for the FC scheme, denoted as $R^{\rm fc}_s$, is given by
\begin{align}\label{eq: R_Lb_1_single}
   R^{\rm fc}_s
   &= \mathbb{E}_{\Zm_{\Kc}^{(1)}} \left[{R}_s^{\rm fc}\left(\Zm_{\Kc}^{(1)}, C_{\cal K}\right)\right].
\end{align}
where  $R_s^{\rm fc} (\Zm_{\cal K}^{(1)}, C_{\cal K})$, is given by the term inside the minimization in 
\eqref{eq:R_ZK} where the variables $r_\Kc$ are set to the values $\{r_k^s\}$ given by 
    \begin{align}\label{eq:r_single_two_user}
        r_k^s = \arg \mathop {\max }\limits_{r_{k} \geq 0} \left\{ \mathop {\min }
        \left\{C_k - R_k(D)- r_k, ~ \log \left[ 1 + \frac{1 - 2^{-r_k}}{2D + \sigma^2}\|\Zm_k^{(1)}\|^2 \right]  \right\}\right\}.
    \end{align}
The expectation in \eqref{eq: R_Lb_1_single} is obtained by MC simulation.

\section{Numerical Results}
\label{simulation}

In this section we present some illustrative results for the single-user and the two-user case. 
 For simplicity, we consider equal capacity constraint, i.e., $C_1 = C_2 = C$.

\subsection{The Single-user Case}

\subsubsection{Evaluation of the Cooperative Informed Receiver Upper Bound}
\begin{figure*}[ht!]
\centering
        \begin{subfigure}[b]{0.45\textwidth}
        \includegraphics[width=1.03\columnwidth]{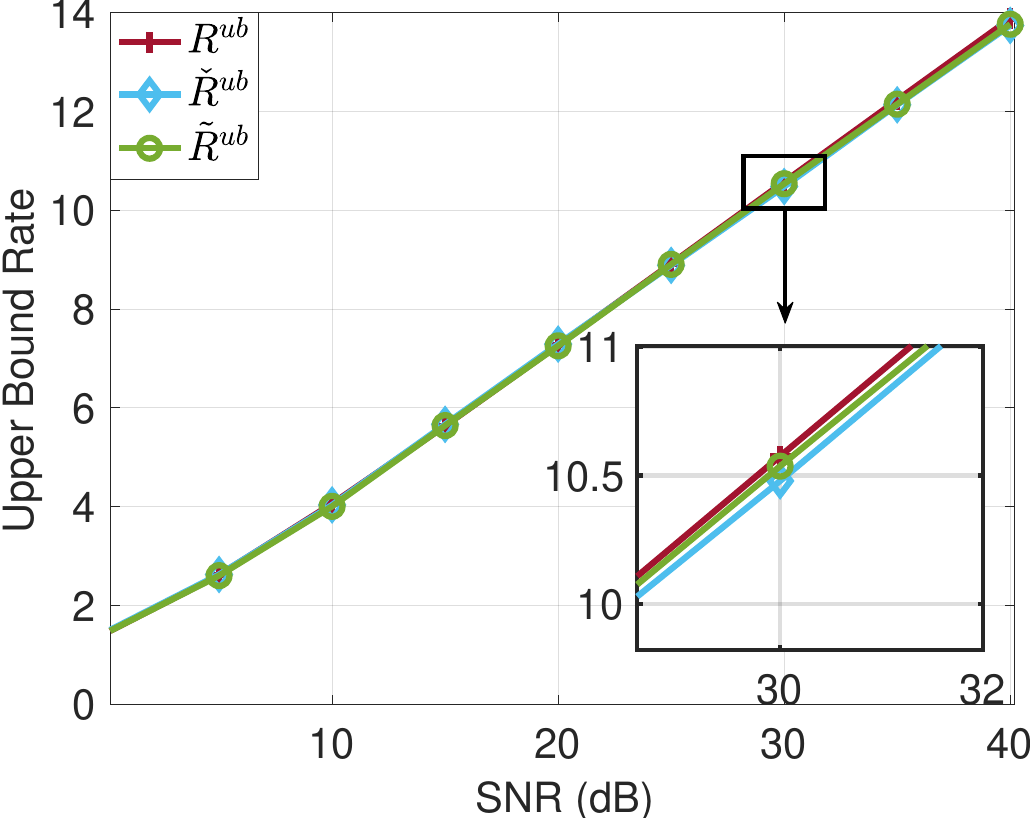}
        \caption{$C = 10$ bits/complex dimension}
        \label{R_SNR}
        \end{subfigure}
        ~~
        \begin{subfigure}[b]{0.45\textwidth}
        \includegraphics[width=\columnwidth]{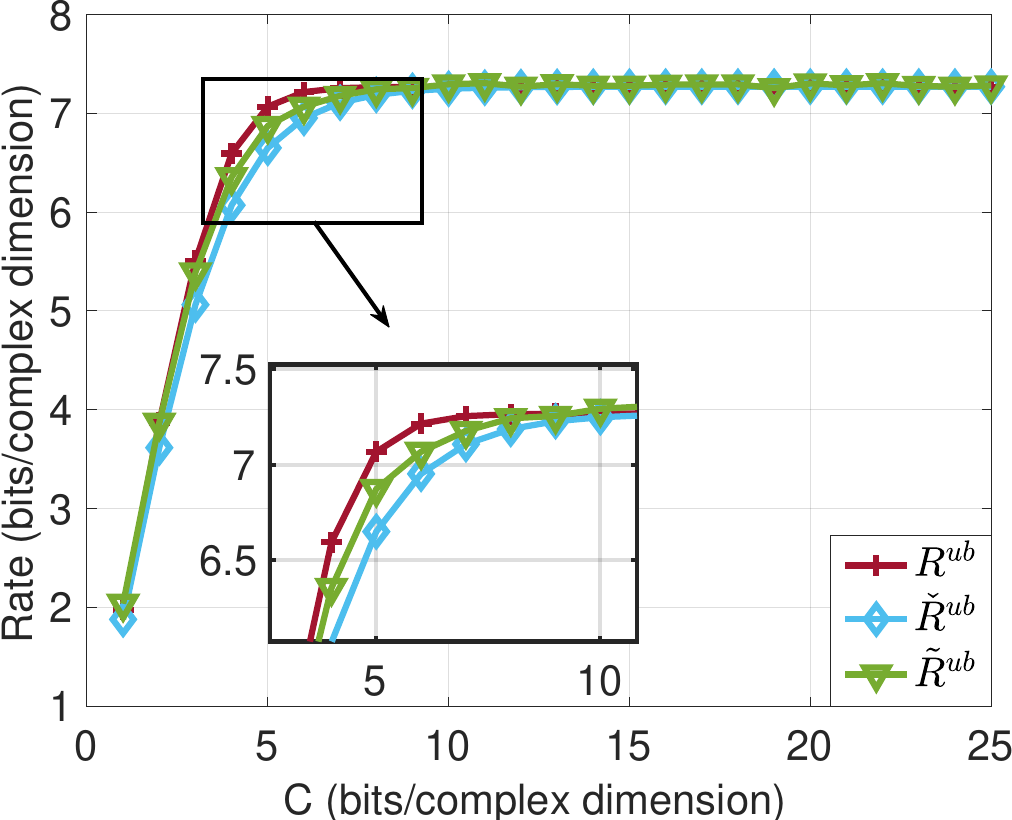}
        \caption{${\rm SNR}=20~$dB}
        \label{R_C}
        \end{subfigure}
 \caption{{Cooperative informed receiver upper bound, lower bound, and approximation, of the informed receiver upper bound $R_0^{\rm ub}$.}}
     \label{fig:baselines}
     \vspace{-6mm}
\end{figure*}

{We consider }
\begin{itemize}
    \item the cooperative informed receiver upper bound $R^{\text{ub}}$ defined in \eqref{R_up_KM};
\end{itemize}
and the following two methods to explicitly evaluate the informed receiver upper bound $R^{\rm ub}_0$:
\begin{itemize}
    \item the lower bound to $R^{\rm ub}_0$, denoted as  $\check{R}^{\text{ub}}$, defined in \eqref{ergodic_problem_aa};
    \item the approximation to $R^{\rm ub}_0$, denoted as  $\widetilde{R}^{\rm ub}$, defined in Algorithm \ref{alg:loop}.
\end{itemize}
 In evaluation of $\widetilde{R}^{\rm ub}$, we choose the parameter $V$ as $100$ and change the value of $C_{\rm max}$ from $C$ to $C + 20$ in step 4 and select the one that yields the highest rate value.
 
In Fig. \ref{R_SNR}, the effect of SNR is investigated with $C = 10$ bits/complex dimension.
Notice that $R^{\text{ub}}$ is an upper bound to $\check{R}^{\text{ub}}$.
Hence, we observe that $R^{\text{ub}}$ is always larger than $\check{R}^{\text{ub}}$ over different SNR values, but the gap is quite small. 
Furthermore, $\widetilde{R}^{\rm ub}$ is also close to $R^{\text{ub}}$ and $\check{R}^{\text{ub}}$. 
In Fig. \ref{R_C}, the effect of the relay-CP link capacity $C$ for fixed SNR is illustrated. Also in this case, 
the values of  $R^{\rm ub}$, $\check{R}^{\text{ub}}$ and $\widetilde{R}^{\rm ub}$ are very close and in the correct order, 
hinting to the fact that the DPP method with the chosen parameters is effectively able to very closely approximate
the informed receiver upper bound $R_0^{\rm ub}$. 

\subsubsection{Achievable Schemes}
\begin{figure*}[t]
\centering
        \begin{subfigure}[b]{0.45\textwidth}
        \includegraphics[width=\columnwidth]{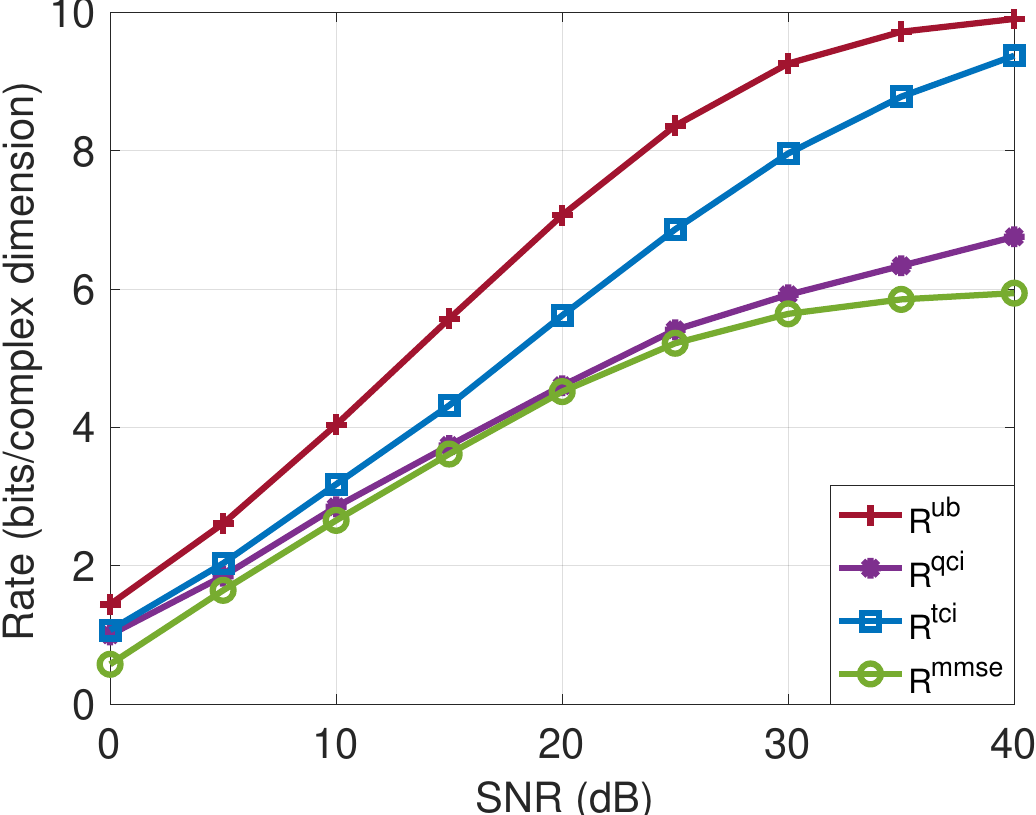}
        \caption{$C = 5$ bits/complex dimension}
        \label{R_VS_rho_1}
        \end{subfigure}
        ~~
        \begin{subfigure}[b]{0.45\textwidth}
        \includegraphics[width=\columnwidth]{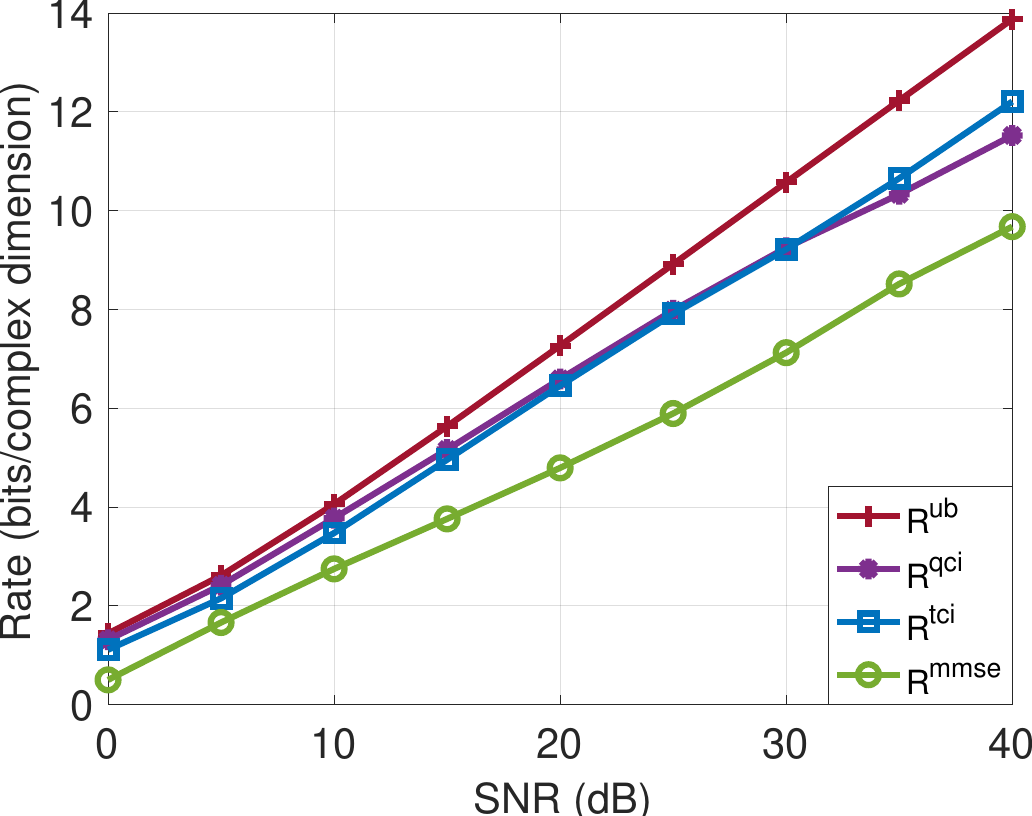}
        \caption{$C = 10$ bits/complex dimension}
        \label{R_VS_rho_2}
        \end{subfigure}
 \caption{The three proposed achievable schemes compared with the
 cooperative informed receiver upper bound for fixed $C$, as a function of the SNR.}
     \label{R_VS_rho}
\end{figure*}

We have evaluated the achievable rates obtained by three proposed different achievable schemes:
\begin{itemize}
    \item quantized channel inversion (QCI) scheme: achievable rate,  $R^{\rm qci}$, defined in \eqref{ergodic_QCI_2};
    \item truncated channel inversion (TCI) scheme: achievable rate, $R^{\rm tci}$, defined in \eqref{ergodic_TCI_2};
    \item MMSE-based scheme: achievable rate, $R^{\rm mmse}$, as defined in \eqref{R_lb_MMSE}.
\end{itemize}
and compared them with the cooperative informed receiver upper bound $R^{\rm ub}$. 
For the QCI scheme, quantization points are chosen with equal probability, i.e., ${\widehat P}_{k, j_k} = \frac{1}{J}$ for all $k \in {\cal K}$ and $j_k \in {\cal J}$, with $J = 2^B$ and we vary the quantization bits $B = [1, 2, 3, 4]$ to yield the maximum value of the rate. For the TCI scheme, we vary the threshold \( S_{\text {th}} \) from 0 to 2 with a step size of 0.1 to find the maximum value of the rate.

In Fig.~\ref{R_VS_rho}, the achievable rates and the upper bound are plotted as functions of SNR, with $C = \{5, 10\}$ bits per complex dimension. 
When \(C = 5\) bits, as SNR increases, both \(R^{\text{ub}}\) and \(R^{\rm tci}\) approach the sum of the two relay CP link capacities, \( C_1 + C_2 \), demonstrating the superior performance of the TCI scheme. This observation is consistent with the limit behavior of $R^{\rm ub}$ in Lemma \ref{re:R_ub}. When \(C = 10\) bits, the convergence of the rates is not clearly visible in the practical SNR range of $0$ to $40$ dB.
However, both \(R^{\rm qci}\) and \(R^{\rm tci}\) closely match the upper bound \(R^{\rm ub}\), highlighting the effectiveness of both the TCI and QCI schemes.

\begin{figure*}[ht!]
\centering
        \begin{subfigure}[b]{0.45\textwidth}
        \includegraphics[width=1.03\columnwidth]{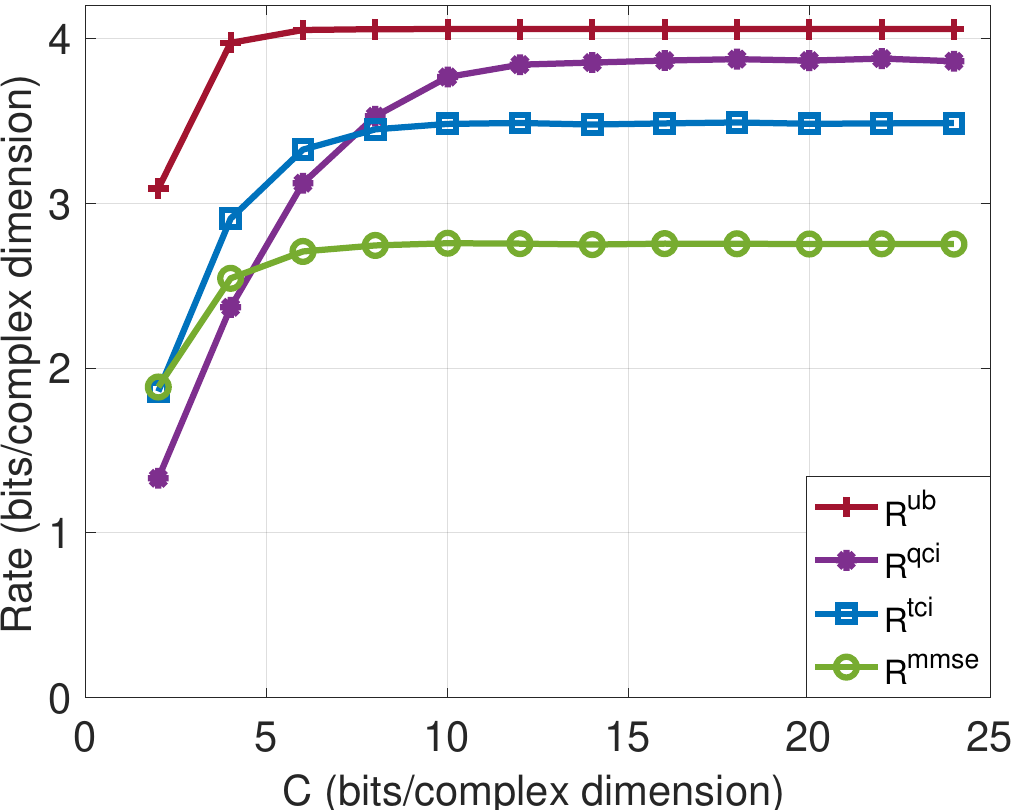}
        \caption{{SNR} $= 10$~dB}
        \label{R_VS_C_1}
        \end{subfigure}
        ~
        \begin{subfigure}[b]{0.45\textwidth}
        \includegraphics[width=\columnwidth]{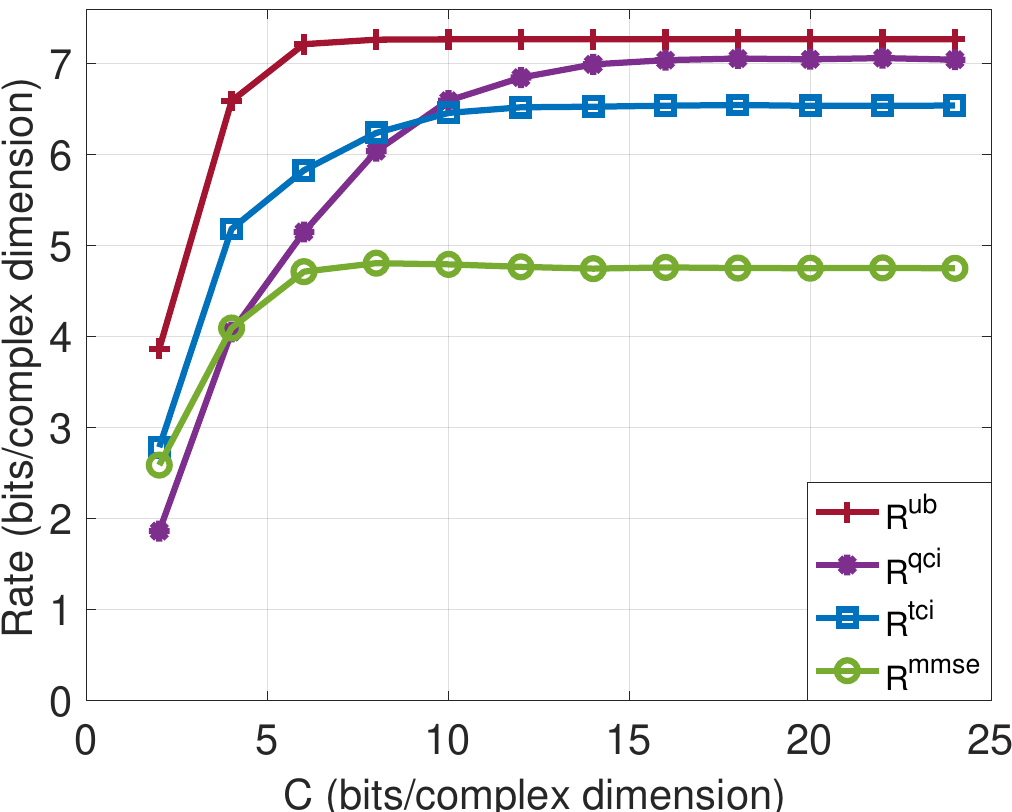}
        \caption{{SNR} $= 20$~dB}
        \label{R_VS_C_2}
        \end{subfigure}
 \caption{{The informed receiver upper bound and achievable rates from three achievable schemes versus $C$.}}
     \label{R_VS_C}
\end{figure*}

In addition, Fig.~\ref{R_VS_C} illustrates the effect of link capacity \( C \) on achievable rates for SNR values of \( \{10, 20\} \) dB. As \( C \) increases, all rates increase monotonically and eventually converge to constant values.
For small \( C \), the gap between \( R^{\rm tci} \) and \( R^{\rm ub} \) is smaller than the gap between \( R^{\rm qci} \) and \( R^{\rm ub} \), emphasizing the superior performance of the TCI scheme in the limited capacity regime. As \( C \) increases,  \( R^{\rm qci} \) performs better than \( R^{\rm tci} \), while both closely approach \( R^{\rm ub} \), demonstrating the superior performance of the QCI scheme when the link capacity is sufficiently large.

\subsection{The Two-user Case}

\subsubsection{Evaluation of the Cooperative Informed Receiver Upper Bound}
\begin{figure*}[t]
\centering
        \begin{subfigure}[b]{0.45\textwidth}
        \includegraphics[width=\columnwidth]{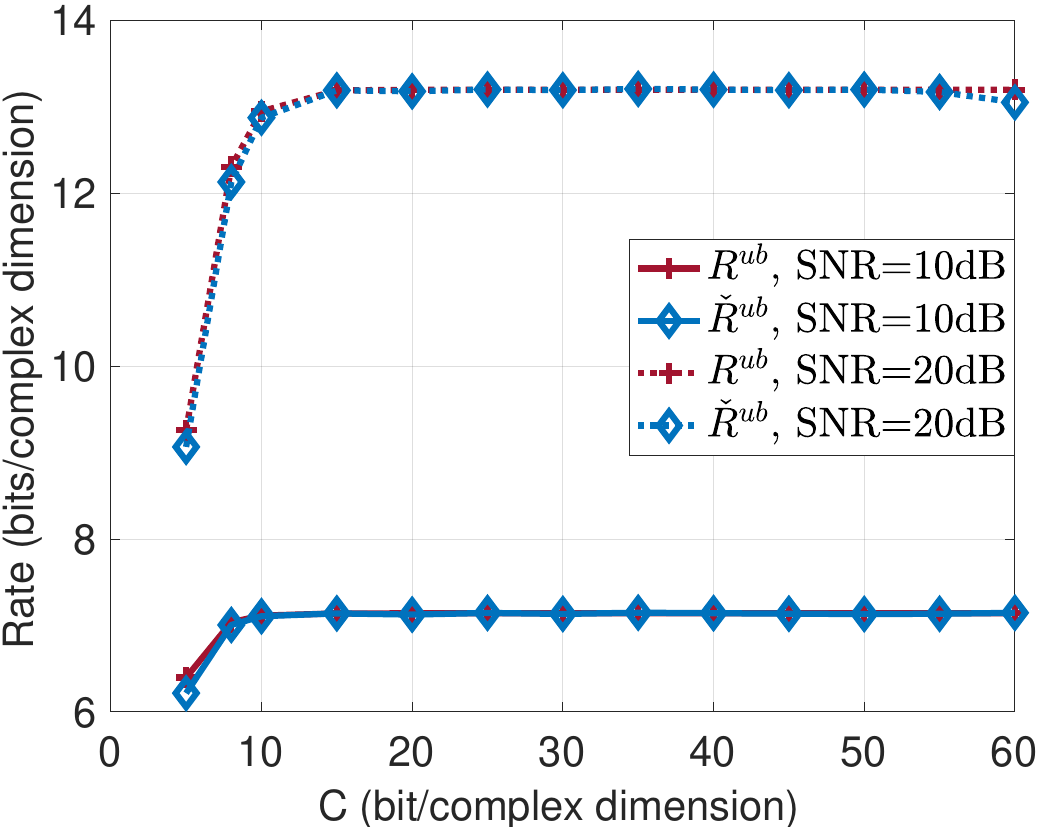}
        \caption{SNR$=\{10, 20\}$ dB}
        \label{two_user_ub_eval_SNR}
        \end{subfigure}
        ~~
        \begin{subfigure}[b]{0.45\textwidth}
        \includegraphics[width=\columnwidth]{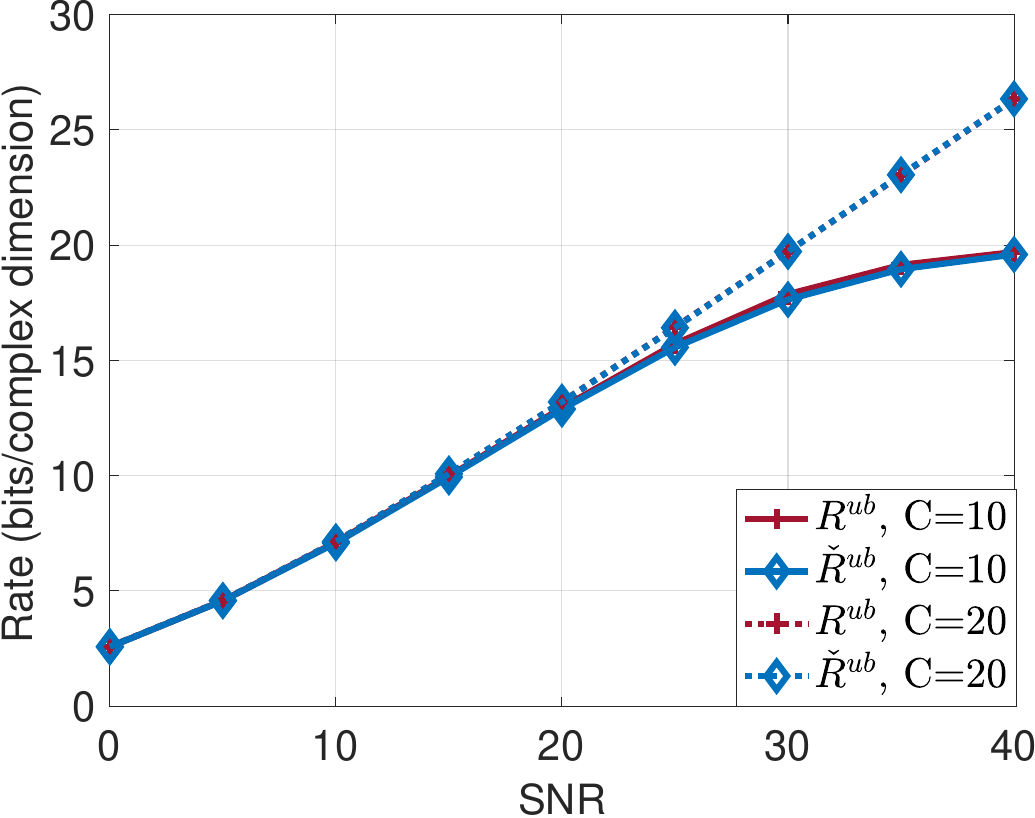}
        \caption{$C = \{10, 20\}$ bits/complex dimension}
        \label{two_user_ub_eval_C}
        \end{subfigure}
 \caption{{The cooperative informed receiver upper bound and the lower bound to the informed receiver upper bound. }}
     \label{two_user_UB}
\end{figure*}

In this subsection, we present simulations for the two-user model, considering the following:
\begin{itemize}
    \item the cooperative informed receiver upper bound, \( R^{\text{ub}} \), as defined in \eqref{R_up_KM_2};
    \item the lower bound of the informed receiver upper bound, \( \check{R}^{\text{ub}} \), as defined in \eqref{ergodic_problem_aa_two_user}.
\end{itemize}
In Fig. \ref{two_user_UB}, we compare the analytical upper bound \( R^{\text{ub}} \) with the lower bound \( \check{R}^{\text{ub}} \) under different SNR and relay-CP link capacities.
As shown in Lemma \ref{R_ub_two_user}, as \( C_1 + C_2 \) approaches infinity, \( R^{\rm ub} \) converges to the expression in \eqref{eq:limit_case_C_inf}, which is a monotonically increasing function of the channel SNR. This explains why the convergence value at 20 dB SNR is higher than that at 10 dB, as seen in Fig. \ref{two_user_ub_eval_SNR}. In addition, as SNR increases to infinity, \( R^{\rm ub} \) converges to \( \sum_{k \in {\cal K}} C_k \), which is also observed in Fig. \ref{two_user_ub_eval_C}.
In both figures, we observe that \( \check{R}^{\text{ub}} \) closely matches \( R^{\text{ub}} \), with a negligible gap over a wide range of system parameters. This confirms that the proposed upper bound is valid and effective in evaluating the proposed achievable schemes.

\subsubsection{Distortion in the Achievable Scheme}
\begin{figure*}[ht!]
\centering
        \begin{subfigure}[b]{0.45\textwidth}
        \includegraphics[width=\columnwidth]{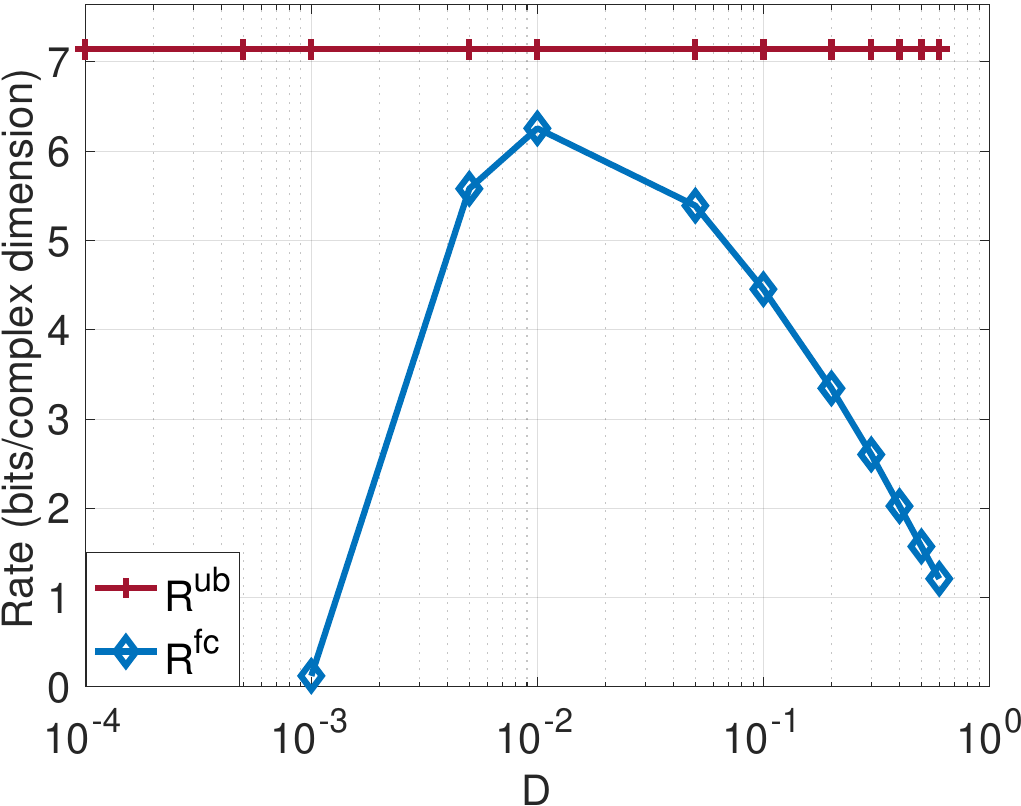}
        \caption{$C = 20$ bits/complex dimension  $\&$ SNR = $10~$dB}
        \end{subfigure}
        ~
        \begin{subfigure}[b]{0.46\textwidth}
        \includegraphics[width=\columnwidth]{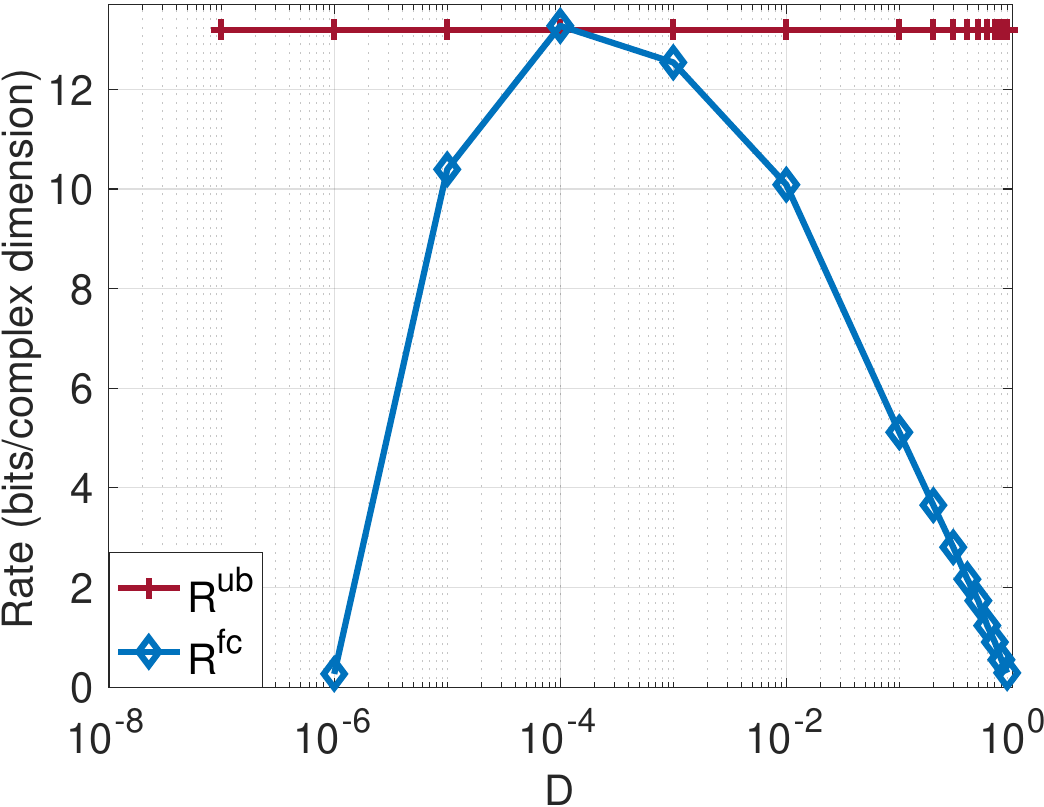}
        \caption{$C = 40$ bits/complex dimension $\&$ SNR = $20~$dB}
        \end{subfigure}
 \caption{{The achievable rate from the {fronthaul compression scheme} versus $D$.}}
     \label{fig:D_func}
\end{figure*}

For the two-user model, we investigate the effect of distortion \( D \) on the proposed achievable scheme for a given SNR and link capacity \( C \).  In the simulation, two scenarios are considered: \( \text{SNR} = 10~\text{dB} \) with \( C = 20 \) bits, and \( \text{SNR} = 20~\text{dB} \) with \( C = 40 \) bits. Fig. \ref{fig:D_func} shows the achievable rates of the fronthaul compression scheme as a function of \( D \), using a base-10 logarithmic scale on the x-axis. 
As can be seen in Fig. \ref{fig:D_func}, the achievable rate is not a monotonic function of \( D \). Instead, there exists an optimal value of \( D \) that maximizes the achievable rate for a given SNR and link capacity \( C \). This demonstrates a trade-off between the accuracy of the channel state representation and the achievable rate over the fronthaul link.

\subsubsection{Achievable Schemes and Cooperative Informed Receiver Upper Bound}

\begin{figure*}[t]
\centering
        \begin{subfigure}[b]{0.45\textwidth}
        \includegraphics[width=\columnwidth]{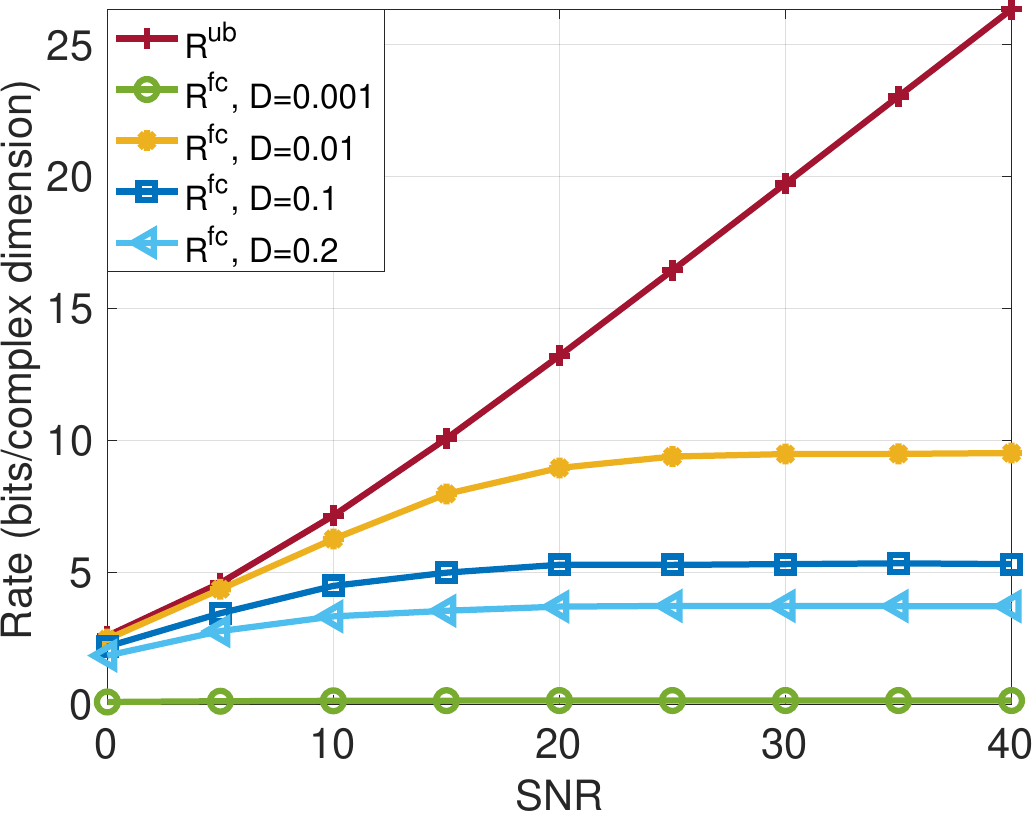}
        \caption{$C = 20$ bits/complex dimension}
        \end{subfigure}
        ~
        \begin{subfigure}[b]{0.45\textwidth}
        \includegraphics[width=\columnwidth]{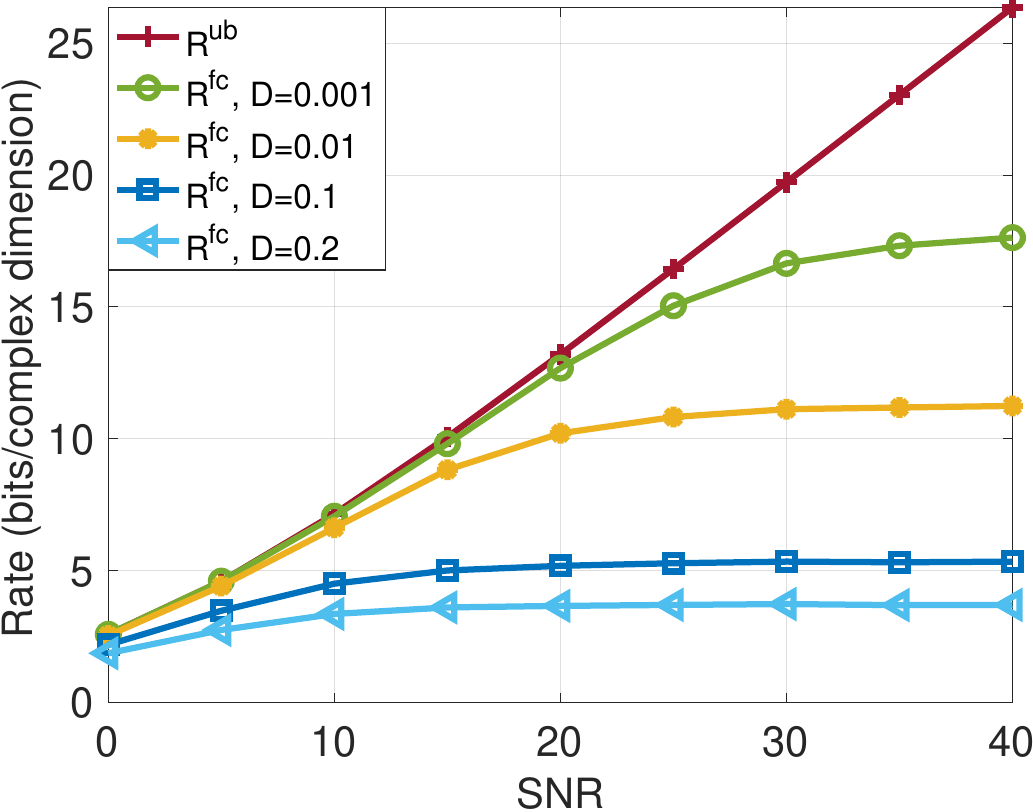}
        \caption{$C = 40$ bits/complex dimension}
        \end{subfigure}
 \caption{{The informed receiver upper bound and the achievable rate from the {fronthaul compression scheme} versus SNR.}}
     \label{R_VS_rho_CRAN}
\end{figure*}

\begin{figure*}[t]
\centering
        \begin{subfigure}[b]{0.45\textwidth}
        \includegraphics[width=\columnwidth]{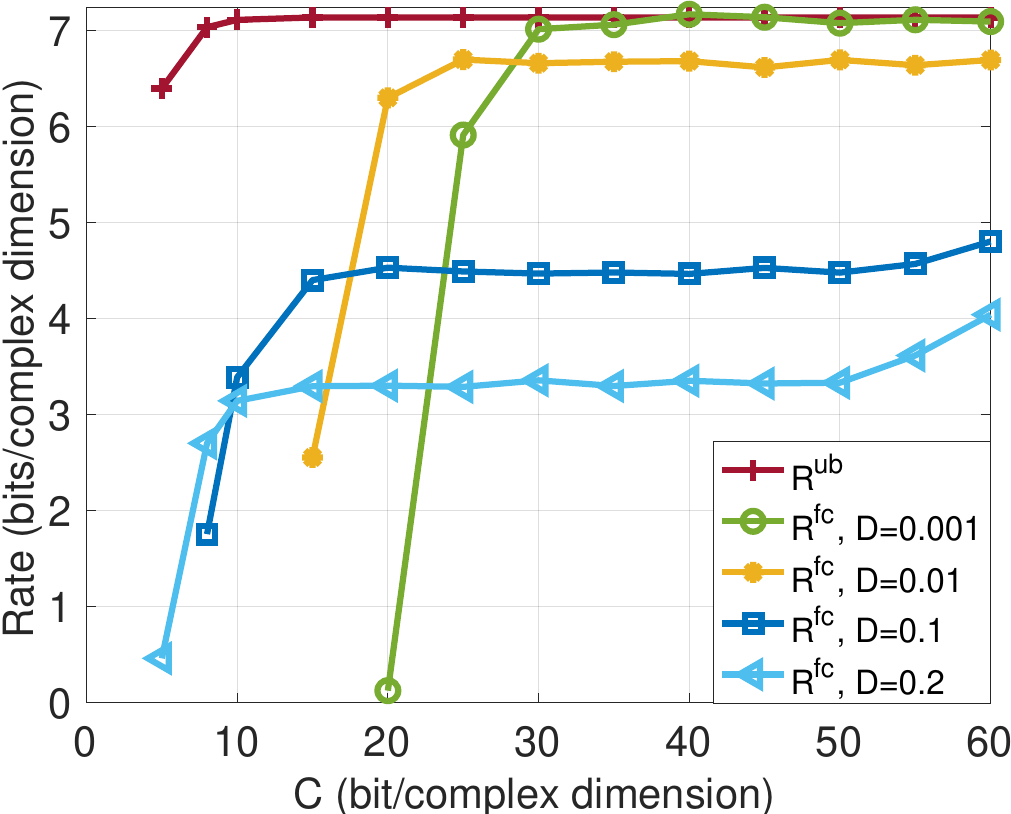}
        \caption{${\rm SNR} =10$~dB}
        \end{subfigure}
        ~
        \begin{subfigure}[b]{0.45\textwidth}
        \includegraphics[width=\columnwidth]{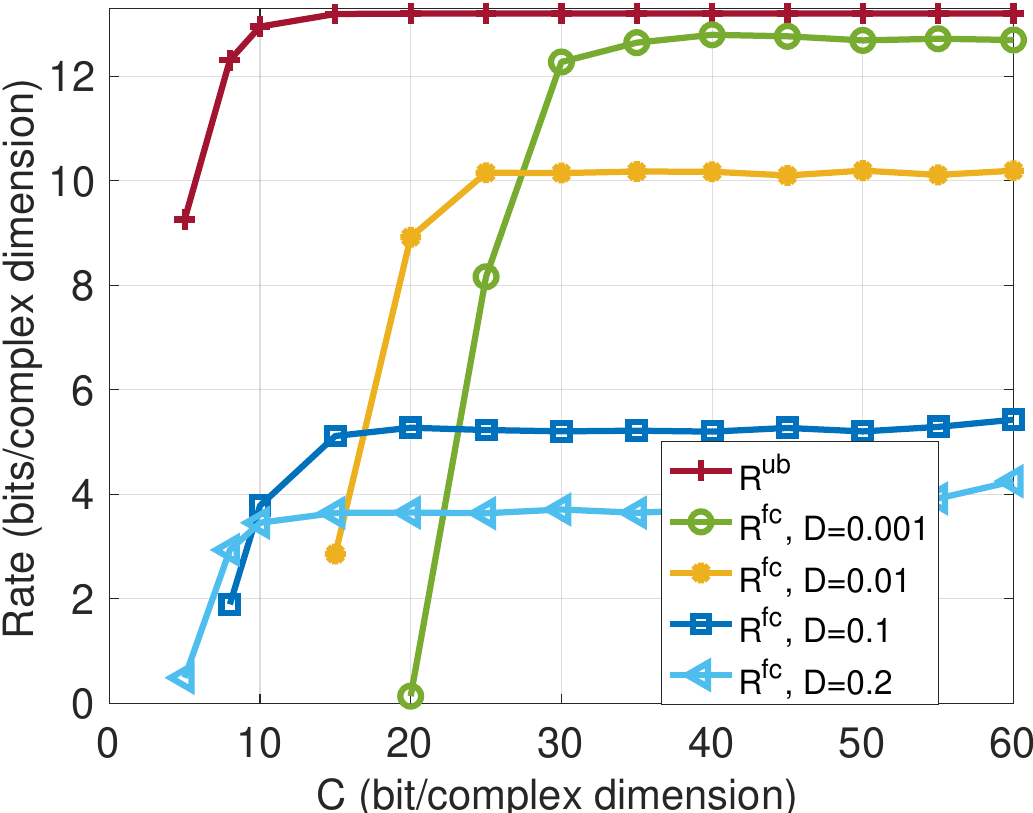}
        \caption{${\rm SNR} =20$~dB}
        \end{subfigure}
 \caption{{The informed receiver upper bound and the achievable rate from the {fronthaul compression scheme} versus link capacity.}}
     \label{R_VS_C_CRAN}
\end{figure*}


Fig. \ref{R_VS_rho_CRAN} compares the cooperative informed receiver upper bound \( R^{\rm ub} \) with the achievable rates from the fronthaul compression scheme \( R^{\rm fc} \), over various distortion levels \( D \) as SNR varies from 0 dB to 40 dB. The link capacities are fixed at \( C = \{20, 40\} \) bits.
As shown, \( R^{\rm ub} \) increases monotonically with SNR in the range of 0 to 40 dB. For different values of \( D \), the achievable rates first increase with SNR and then plateau at a horizontal asymptote determined by \( C \), \( D \), and SNR. In the limited-capacity regime, such as \( D = 0.001 \) and \( C = 20 \) bits, the low distortion requires a significant portion of the available capacity allocated to encode the channel state. This allocation completely consumes the link capacity, leaving no capacity for signaling and resulting in a zero rate.
As \( C \) increases to $40$ bits for the same distortion (\( D = 0.001 \)), the capacity allocated to encode the channel state remains unchanged, but the larger total capacity allows more bits to be allocated for signal transmission, resulting in higher achievable rates. For sufficient link capacities (\( C = 40 \) bits), \( R^{\rm fc} \) for \( D = 0.001 \) approaches \( R^{\rm ub} \) over a wide range of practically relevant SNR values. This illustrates the trade-off between the accuracy of channel state information and the efficiency of signal transmission. 
It is also important to note that when link capacity becomes the bottleneck, increasing the SNR does not result in a proportional increase in the achievable rate.

Fig. \ref{R_VS_C_CRAN} further analyzes the effect of link capacity \( C \), showing \( R^{\rm ub} \) alongside the achievable rate for the fronthaul compression scheme \( R^{\rm fc} \) for SNR values of \( \{10, 20\} \) dB as \( C \) ranges from $5$ to $60$ bits. In the low-capacity regime, limited link capacity restricts the number of bits available for channel state encoding, resulting in significant channel state quantization errors and creating a bottleneck in the achievable rate. 
In this regime of low fronthaul rate, our achievable scheme is quite far from the upper bound.
This shows that there is potentially some significant gain to be made by devising more refined achievable schemes. 
As \( C \) increases, more bits can be allocated for channel state compression, enhancing the achievable rates. When the link capacity is sufficiently large ($C = 40$ bits), \( R^{\rm fc} \) can approach \( R^{\rm ub} \) with an appropriate choice of distortion (\( D \)). This illustrates that the FC scheme becomes very effective when 
the rate constraint in the fronthaul is not the system bottleneck.

\section{Conclusions}
\label{conclusion}

In this paper, we studied the fundamental limits of a disaggregated radio access network where relays operate without knowledge of users' codebooks and the central processor lacks channel state information. The focus was on understanding the achievable rates under these constraints in a setting with either one or two users and two relays. For the single-user case, we derived analytical upper bounds assuming genie-aided access to channel state at the central processor and proposed three practical achievability schemes based on local relay processing. These included strategies using quantized and truncated channel inversion as well as minimum mean square error estimation. For the two-user scenario, where the challenges increase due to the inability to invert the channel at the relays, we extended the analysis by deriving new upper bounds and proposing a novel achievable scheme that jointly compresses the received signal and the channel state. Across both settings, we demonstrated through simulations that the achievable schemes perform close to the respective upper bounds over a broad range of system parameters. The results highlight that significant rate gains can be achieved even with only partial channel state information and local processing, underscoring the potential of simple, distributed strategies for future wireless architectures constrained by fronthaul capacity and partial observability.

While this work addresses the sum capacity characterization under the assumption of i.i.d. channel states known only at the relays, several natural extensions remain open and need further exploration. 
One compelling research direction involves relay cooperation through finite-capacity links. This setting introduces new challenges and opportunities, where results such as those developed in \cite{dikshtein2022bounds} become relevant. Another important extension relaxes the i.i.d. assumption on the channel state sequence. Specifically, if the state remains fixed for a block of \( T \) channel uses and is known only at the relays, the problem interpolates between the setting studied in this paper (\( T=1 \)) and a regime of \( T \to \infty \) where there is no cost to convey the channel state to the CP. 
This provides rich connections to the framework of the broadcast approach as discussed in \cite{tajer2021broadcast}.
Another interesting extension is the case of dependent fading coefficients across users, which occurs when they operate in similar physical environments, resulting in correlated time-varying channels.

\appendices
\section{Proof of Lemma \ref{re:R_ub}}
\label{R_ub}
According to \eqref{bottle_constr_KM}, $\int_{\nu \sigma^2}^{\infty} \left( \log \frac{\lambda}{\nu \sigma^2} \right) f_\lambda (\lambda) d \lambda$ increases as $\nu$ decreases. When SNR approaches infinity, i.e., $\sigma \rightarrow 0$, to ensure that \eqref{bottle_constr_KM} is satisfied, $\nu$ becomes large. Then, we have 
\begin{align}
    R^{\rm ub} &=  \int_{\nu \sigma^2}^{\infty} \left[ \log \left(1 + \frac{\lambda}{\sigma^2} \right) - \log (1 + \nu)\right] f_\lambda (\lambda) d \lambda \nonumber \\
    & \rightarrow \int_{\nu \sigma^2}^{\infty} \left[ \log \left(\frac{\lambda}{\sigma^2} \right) - \log (\nu)\right] f_\lambda (\lambda) d \lambda \nonumber \\
    & = \sum_{k \in {\cal K}} C_k.
\end{align}
Therefore, when SNR approaches infinity, $R^{\rm ub}$ converges to $\sum_{k \in {\cal K}} C_k$.

Next, when SNR is fixed, and $\sum_{k \in {\cal K}} C_k$ approaches infinity, from \eqref{bottle_constr_KM}, it can be found that $\nu$ is $0$. Therefore, based on \eqref{R_up_KM}, the upper bound approaches 
\begin{align}
    R^{\rm ub} \rightarrow \int_{0}^{\infty} \left[ \log \left(1 + \frac{\lambda}{\sigma^2} \right)\right] f_\lambda (\lambda) d \lambda.
\end{align}

\section{Derivation of Theorem \ref{th:optimal_analysis}}
\label{sec:DDP_theorem_proof}

We adopt the classical Lyapunov DPP framework \cite{neely2010stochastic, neely2014simple} to solve problem~ \eqref{ergodic_problem_LB}. The solution approach consists of the following three components, evaluated iteratively at each time slot \( t \):

\begin{itemize}
    \item[1)] \textbf{Allocated instantaneous link capacities:} Observe the virtual queue states \( Q_k(t) \), and determine the virtual arrivals \( \{c_k(t, \rho_\Kc(t))\} \) as the optimal value of the $\{c_k\}$ variables in the optimization problem:
    \begin{subequations}\label{ergodic_problem_drift}
        \begin{align}
            \mathop{ \min}_{c_{\mathcal{K}}} \quad & \sum_{k \in \mathcal{K}} c_k Q_k(t) - V R(\rho_{\mathcal{K}}(t), c_{\mathcal{K}}) \\
            \text{s.t.} \quad & 0 \leq c_k \leq C_{\max}, \quad \forall k \in \mathcal{K},
        \end{align}
    \end{subequations}
    where \( R(\rho_{\mathcal{K}}(t), c_{\mathcal{K}}) \) is defined in \eqref{R_fixed_rho}. 
As we have seen, this problem can be reformulated as the convex program in \eqref{eq:solve_problem} by introducing an auxiliary variable \( \beta \), and yields \eqref{eq:solve_problem}.
    
    \item[2)] \textbf{Instantaneous rate:} The optimal instantaneous rate \( \beta^*(t) \) is obtained as the solution to the reformulated problem \eqref{eq:solve_problem}, corresponding to the optimal value of the variable $\beta$.

    \item[3)] \textbf{Virtual Queue Update:} The virtual queues $Q_k(t+1)$ evolve according to the update rule in \eqref{eq:virtual_queue}.
\end{itemize}

Let \( Q_{\Kc}(t) = \{Q_1(t), Q_2(t)\} \) denote the virtual queue backlogs and
define the quadratic Lyapunov function as  
\[
    L(Q_\Kc(t)) = \frac{1}{2} \sum_{k \in \mathcal{K}} Q_k^2(t),
\]  
and the corresponding conditional Lyapunov drift as  
\[
    \Delta(Q_\Kc(t)) = \mathbb{E} \left[ L(Q_\Kc(t+1)) - L(Q_\Kc(t)) \,|\, Q_\Kc(t) \right].
\]
To bound the Lyapunov drift, observe that from \eqref{eq:virtual_queue} and the inequality \( \max(x, 0)^2 \leq x^2 \), we have:
\begin{align}\label{eq:bound_dif_Q_t}
    Q_k(t+1)^2 - Q_k(t)^2 
    &\leq \left( Q_k(t) + c_k(t, \rho_{\mathcal{K}}(t)) - C_k \right)^2 - Q_k^2(t) \nonumber \\
    &= 2 Q_k(t) \left( c_k(t, \rho_{\mathcal{K}}(t)) - C_k \right) + \left( c_k(t, \rho_{\mathcal{K}}(t)) - C_k \right)^2.
\end{align}

Summing over \( k \in \mathcal{K} \) and taking expectations, we obtain an upper bound on the conditional drift:
\begin{align}\label{eq:delta_theta_t}
    \Delta(Q_\Kc(t)) 
    &\leq \mathbb{E} \left[ \sum_{k \in \mathcal{K}} (c_k(t, \rho_{\mathcal{K}}(t)) - C_k) Q_k(t) \,|\, Q_\Kc(t) \right] + \frac{1}{2} \mathbb{E} \left[ \sum_{k \in \mathcal{K}} \left( c_k(t, \rho_{\mathcal{K}}(t)) - C_k \right)^2 \,|\, Q_\Kc(t) \right].
\end{align}
By Theorem 4.2 in \cite{neely2010stochastic}, the time-averaged rate achieved by the drift-plus-penalty 
algorithm is within \( O(1/V) \) of the optimal value. Specifically, we have the bound:
\begin{align}
    \lim_{T \rightarrow \infty} \frac{1}{T} \sum_{t=0}^{T-1} \beta^*(t) \geq R^{\rm ub}_0 - \frac{B}{V},
\end{align}
where \( R^{\rm ub}_0 \) denotes the optimal objective value of problem \eqref{ergodic_problem_LB}, 
and \( B \) is a finite constant that upper bounds the second term in the RHS of \eqref{eq:delta_theta_t}.
From the constraint \( 0 \leq c_k \leq C_{\max} \), an upper bound to the second term in the RHS of \eqref{eq:delta_theta_t} is given by
\begin{align}\label{eq:B}
    \frac{1}{2} \mathbb{E} \left[ \sum_{k \in \mathcal{K}} (c_k(t, \rho_{\mathcal{K}}(t)) - C_k)^2 \,|\, Q_\Kc(t) \right] 
    &\leq \frac{1}{2} \sum_{k \in \mathcal{K}} \max \left\{ C_k^2, (C_{\max} - C_k)^2 \right\} \nonumber \\
    &\overset{\Delta}{=} B.
\end{align}

\bibliographystyle{IEEEtran}
\bibliography{IEEEabrv,Ref,refs-CF-UC}
\end{document}